\documentclass[11pt,a4paper]{article}
\usepackage{amsthm}
\usepackage{amsmath}
\usepackage{amssymb}
\usepackage{dsfont}
\begin{document}
\title{Quantum correlations; quantum probability approach.} 
\author{W\l adys\l aw Adam Majewski}


\date{}


\newcommand{\Addresses}{{
  \bigskip
  \footnotesize

  \textsc{Institute of Theoretical Physics and Astrophysics, Gda{\'n}sk  University, Wita Stwo\-sza~57, 80-952 Gda{\'n}sk, Poland and Unit for BMI, North-West-University, Potchefstroom, South Africa}\par\nopagebreak
  \textit{E-mail address:} \texttt{fizwam@univ.gda.pl}

}}


\theoremstyle{plain}
\newtheorem{thm}{Theorem}[section] 
\theoremstyle{definition}
\newtheorem{defn}[thm]{Definition} 
\newtheorem{rem}[thm]{Remark}
\newtheorem{prop}[thm]{Proposition}
\newtheorem{exmp}[thm]{Example} 
\newtheorem{rle}[thm]{Rule}
\newtheorem{fc}[thm]{Fact}
\newtheorem{cor}[thm]{Corollary}
\newtheorem{proble}[thm]{Problem}
\newtheorem{con}[thm]{Conjecture}
\newtheorem{pro}[thm]{Proposition}
\newcommand{\Tr}{\mathrm{Tr}}
\newcommand{\Cn}{{\setbox0=\hbox{
$\displaystyle\rm C$}\hbox{\hbox
to0pt{\kern0.6\wd0\vrule height0.9\ht0\hss}\box0}}} 
\newcommand{\Rn}{{\rm I\!R}}
\maketitle

\begin{abstract}
This survey gives a comprehensive account of quantum correlations understood as a phenomenon stemming from the rules of quantization. Centered on quantum probability it describes the physical concepts related to correlations (both classical and quantum), mathematical structures, and their consequences. These include the canonical form of classical correlation functionals, general definitions of separable (entangled) states, definition and analysis of quantumness of correlations, description of entanglement of formation, and PPT states. This work is intended both for physicists interested not only in collection of results but also in the mathematical methods justifying them, and mathematicians looking for an application of quantum probability to concrete new problems of quantum theory.
\end{abstract}

\newpage
\tableofcontents
\newpage
\section{Introduction}
The notion of correlations is a well established concept in probability theory, see \cite{H}, \cite{B}. This concept has been successfully employed in classical physics to describe the size and direction of a relationship between two or more (classical) observables. Moreover, it  is obvious (in the framework of classical physics) that a correlation between observables does not automatically mean that a change occurred in one observable is the result of a change appeared in the other observable. On the other hand, the causality means that one event is the result of the occurrence of the another one. In particular, this means, that the second event is appearing later. Thus, to describe such situation, in particular, one should also take into account a time evolution as well as to encode the idea of a distance. The latter is indispensable for any discussion of any propagation of effects. 

Consequently, although correlations and causality are related ideas, these two concepts are different and should not be confused. As this difference is a fundamental feature of the probabilistic description of physical systems, the quantization should respect it as well. Thus, speaking about quantum correlations we will not treat causal relations.

The aim of this review is to describe in details how to adapt the probabilistic scheme for a definition and description of correlations in quantum theory. To this end, we begin with a brief review of the classical theory. The important feature of the classical theory is the one-to-one correspondence between (positive) measures and (positive) functionals. More precisely, let $X$ be a compact Hausdorff space. Denote by $C(X)$ the space of all complex valued continuous functions defined on $X$. $C(X)$ equipped with the norm $\| f\| = \sup_{x \in X} |f(x)|$ is a commutative $C^*$-algebra. Let $\phi: C(X) \to \Cn$ be a continuous linear functional. We say it is positive
if $\phi(f)\geq 0$ for $f\geq 0$, and it is normed if $\phi(1_{X})=1$
where $1_{X}$ stands for the identity function on $X$.
For any (positive) functional $\phi$ there exists a (positive) Borel measure $\mu$ on $X$ such that
\begin{equation}
\label{1}
\phi(f) = \int_X f(x) d\mu(x)
\end{equation}
and the correspondence given by (\ref{1}) is one-to-one. Consequently, the measure theoretic description can be ``translated'' to that expressed in terms of functionals. We emphasize that only the ``translated'' approach to probability theory can be quantized in such a form which can be used for a  discussion of the concept of quantum correlations. This is due to the observation that the quantization implies the replacement of the commutative $C^*$-algebra $C(X)$ with a non-commutative one $\mathfrak{A}$. Note, that for a non-commutative $C^*$-algebra the relation (\ref{1}) is not valid. Moreover,
there does not exist a precise notion of a quantum measure on a ``quantum'' space. However, linear positive functionals on $\mathfrak{A}$ still have the well defined meaning. Consequently, the pair $(\mathfrak{A}, \phi)$ is the starting point of quantum probability. To sum up, to ``translate'' a description of classical correlations to quantum ones, we should use the non-commutative counterpart of functionals given by (\ref{1}).

Speaking about quantization, it is worth pointing out that one can embed a commutative $C^*$-algebra into a non-commutative one BUT it is impossible to embed a non-commutative $C^*$-algebra into a commutative one. Consequently, using the above form of quantization we rule out hidden variables model. This explains why we will not be interested in hidden variable models as well as in any discussion concerning such models in our description of quantum probability.

Another important point to note here is that related to the so called Bell inequalities. Namely, we will define and study quantum correlations following the scheme provided by quantization of classical probability approach. Therefore, Bell inequalities may appear as an ingredient of the presented approach. Consequently, Bell inequalities will not be a starting point for a definition of quantum correlations. In other words, a violation 
of Bell inequalities will not be taken as a definition of quantum correlations. 

The paper will be organized as follows. In the next section we have compiled some basic facts from the classical probability theory with an emphasize on certain properties of classical measures. In Section 3, we will discuss classical composite systems. We will show that any state of such a system is a separable one. Then, in Section 4, to be prepared for the quantization procedure, we provide an exposition of rules of quantum mechanics. Section 5 contains a brief summary on the theory of tensor products which is indispensable for a description of quantum composite systems. These systems are studied in Section 6. The next section, Section 7, summarizes without proofs the relevant material on the decomposition theory. The necessity of that part follows from the well known fact that contrary to the classical case, in quantum theory the set of states does not form a simplex. Sections 5, partially 6, and 7 are ``rather mathematical''.
But, they constitute sufficient preparation for understanding various measures of quantum correlations.
On the other hand, the idea of employment of abstract mathematics for description of new concepts in Quantum Theory is not 
new. Contrary, it is even Dirac's advice, see page viii in \cite{dirac}: \textit{ Mathematics is the tool specially suited for dealing with abstract concepts of any kind and there is no limit to its power in this field. For this reason a book on the new physics, if not purely descriptive of experimental work, must be essentially mathematical.}
We only wish to add that quantum entanglement is a part of new physics.

Quantum correlations are defined and studied in Section 8 while in Section 9 and Section 10 entanglement of formation and PPT states, respectively, are studied. Section 11 will be devoted to the problem of describing the time evolution of quantum correlations. An illustrative example of evolution of quantum correlations will be provided. The last section contains final remarks and conclusions.

\section{Probability theory}

The purpose of this section is to review some of the standard facts from probability theory. We begin with basic definitions.

\begin{defn}Let $\Omega$ be a set. $\mathcal{F}$ is a $\sigma$-algebra if $\mathcal{F} \subset 2^\Omega$ and
\begin{enumerate}
\item $\Omega\in \mathcal{F},$
\item if $A\in \mathcal{F}$ then $\Omega \setminus A\in \mathcal{F},$
\item if $A_i\in \mathcal{F}$ for $i=1, 2\ldots $, then $\bigcup_{i=1}^{\infty}A_i\in \mathcal{F}$.\newline
The pair $(\Omega, \mathcal{F})$ is called a measure space.
\end{enumerate}\end{defn}\noindent
and

\begin{defn}A probability measure $p$ on $(\Omega,\mathcal{F})$ is a function $p:\mathcal{F}\rightarrow \left[0,1\right]$ such that
\begin{enumerate}
\item $p(\Omega)=1$

\item $p(\bigcup_{i=1}^{\infty}A_i)=\sum_{i=1}^{\infty} p(A_i)$ if $A_i\in \mathcal{F}$ for $i=1, 2\ldots $,\newline
 and $A_i\cap A_j=\emptyset$ for $i\neq j $.
\end{enumerate}
\end{defn}\noindent
and finally

\begin{defn}A probability space is a triple $(\Omega, \mathcal{F}, p)$ where $\Omega$ is a space (sample space), $\mathcal{F}$ is a $\sigma$-algebra (a family of events), and $p$ is a 
probability measure on $(\Omega, \mathcal{F})$.\end{defn}

\begin{rem}
In probability theory, an ``elementary event'' is considered as the principal undefined term, see \cite{H}. Intuitively speaking, an elementary event has the meaning of the possible outcome of some physical experiment.
\end{rem}

To illustrate these definitions we give:

\begin{exmp}\textit{Discrete probability space.}\newline
 Let $\Omega$ be a countable (finite) set. Let us take
      $\mathcal{F}=2^\Omega$. We put, for $A \subset \Omega$, $p(A)=\sum_{\omega\in A}q(\omega)$, where $q(\omega) \in [0,1]$, $\omega \in \Omega$, are numbers such that $\sum_{\omega \in \Omega} q(\omega) = 1$.
 \newline Clearly, $(\Omega, \mathcal{F}, p)$ is a probability space.
\end{exmp}\noindent
Our next example is:

\begin{exmp} \textit{Continuous probability space.}\newline
\label{2.6}
To be specific we put $\Omega=[0,1]\subset \mathbb{R}$, $\mathcal{F}$ to be a $\sigma$-algebra of Borel sets $\cal B$ in $[0,1]$, and $\lambda$ to be the Lebesgue measure. 
\newline Obviously,  $(\Omega, \mathcal{F}, \lambda)$ is a probability space.
\end{exmp}

In probability theory, the notion of stochastic variable can be considered as a counterpart of a (classical) observable. It is defined as follows:

\begin{defn}
Let $(\Omega, \mathcal{F}, p)$ be a probability space.
 A measurable, real valued function $f:\Omega\rightarrow \mathbb{R}$ is called a stochastic variable.
\end{defn}

\begin{rem}
\label{2.8}
A stochastic variable, frequently also called a random variable, is a function whose numerical values are determined by chance, see \cite{H}. In ``physical'' terms, a stochastic variable can be considered as a function attached to an experiment in such a way that once an experiment has been carried out the value of the function is known.
For the sake of comparison, note the similarity of that feature of (classical) stochastic variable with the corresponding property of an
observable in quantum mechanics. This and the impossibility of embedding a non-commutative $C^*$-algebra into a commutative one explains why we will not discuss contextuality ideas.
\end{rem}

To speak about \textit{correlations}, it is convenient to begin with the opposite notion  - with the notion of \textit{independence}.
 
 \begin{defn}
 Let $(\Omega, \mathcal{F}, p)$ be a probability space. 
We say that two events $A,B\in \mathcal{F}$ are \textbf{independent} if
\begin{equation}
 p(A\cap B)=p(A)p(B)\label{s1eq1}
\end{equation}or more generally, the events
 $A_i$, $i=1,...,n$ in $\mathcal F$ are independent if
\begin{equation}
 p(A_{i_1}\cap A_{i_2} \cap \ldots \cap A_{i_k})=\prod_{j=1}^{k}p(A_{i_j})\label{s1eq2}
\end{equation}
for every $k=1,...,n$ and $i_1,...,i_k$ such that $1\leq i_1<i_2,...,i_k \leq n$.
\end{defn}

It is worth pointing out that
the concept of independence can be extended for stochastic variables.
We say that stochastic variables \newline
 $X_1,\ X_2,\ \ldots, X_n$ are independent if and only if
events $(X_i \in B_i)$, $i=1,...,n$, are independent
where
\begin{equation}
(X_i\in B_i)\equiv \left\{\omega: X_i(\omega)\in B_i\right\} \label{s1eq4}
\end{equation}
for arbitrary Borel sets $B_i\subset \mathbb{R}$.

Further, we can define an expectation value $E(X)$ of a stochastic variable $X$ as 
\begin{equation}
E(X)=\int Xdp \equiv \int X(\omega)dp(\omega)\label{s1eq5}
\end{equation}
We interpret $E(X)$ as the mean value of the stochastic variable $X$ provided that the probability of events is given by the probabilistic measure $p$.

\begin{defn}
Let two stochastic variables $X$, $Y$ be given. Moreover, let their second moments be finite, i.e. $E(X^2)$, $E(Y^2) < \infty$. 
We say that two stochastic variables $X$ and 
$Y$ are \textbf{uncorrelated} if 
\begin{equation}
E(XY)=E(X)E(Y)\label{s1eq6}
\end{equation}
\end{defn}

It is easy to check that:

\begin{equation}
(E(XY)-E(X)E(Y))^2\leq(E(X^2)-E(X)^2)(E(Y^2)-E(Y)^2) \label{s1eq7}
\end{equation}

We are thus led to the following definition, cf \cite{H}:
\begin{defn}
\label{2.11}
A \textbf{correlation coefficient} $C(X,Y)$ is defined as
\begin{equation}
C(X,Y)=\frac{E(XY)-E(X)E(Y)}{(E(X^2)-E(X)^2)^\frac{1}{2}(E(Y^2)-E(Y)^2)^\frac{1}{2}}\label{s1eq8}
\end{equation}
\end{defn}

Let us briefly comment on this definition. Firstly,
taking into account inequality (\ref{s1eq7}), it is obvious that $C(X,Y)\in [-1,1]$. Secondly, if $C(X,Y)$ is equal to $0$ then stochastic variables 
$X$ and $Y$ are \textit{uncorrelated}. Further, if 
$C(X,Y)\in (0,1]$, then $X$, $Y$ are said to be \textit{correlated} and finally when $C(X,Y)\in [-1,0)$, stochastic variables 
$X$ and $Y$ are said to be \textit{anti correlated}.

The above introduced notions for stochastic variables, independence and uncorrelatedness, are related to each other. Namely, cf. Section 45 in \cite{H}.
\begin{pro}
($\Omega=[0,1]$, $\cal F$, $p$) be a probability space. Assume that
stochastic variables are independent and integrable. Then they are also uncorrelated.
\end{pro}

However, these two notions are not equivalent. It is illustrated by an example taken from the Halmos book \cite{H}, Section 45.
\begin{exmp}
Let ($\Omega=[0,1]$, $\cal F$, $\lambda$) be the probability space given in Example \ref{2.6}. Define stochastic variables $f$ and $g$ as
$f(x)=\sin{2\pi x}$ and $g(x)=\cos{2\pi x}$.  
Then the expectation value of $f\cdot g$ is given by $E(fg)=\int_0^1\sin{2\pi x}\cos{2\pi y}dx=0$.
Hence , they are uncorrelated.
On the other hand,
let us define $D=[0,\epsilon)$ where $\epsilon$ is sufficiently small.\newline 
$\lambda(\left\{x:\sin{2\pi x}\in D\right\}\cap \left\{x:\cos{2\pi x}\in D\right\})=0$ while
$\lambda(\left\{x:\sin{2\pi x}\in D\right\})\neq0$ and
$\lambda(\left\{x:\cos{2\pi x}\in D\right\})\neq0$. Therefore $f$ and $g$ are not independent (but they are uncorrelated!).
\end{exmp}

Having clarified the concept of correlations we wish to close this section with a basic package of results as well as the terminology used in measure theory. We wish to quote a result describing an approximation property of a positive measure, see \cite{Bour}, \cite{choquet} vol. I.,  and \cite{Mey}. To this end, for the convenience of the reader we repeat the relevant material from \cite{choquet} without proofs, thus making our exposition self-contained.
Let $E$ be a locally compact Hausdorff space.
A positive Radon measure is a positive linear map $\phi: C_{\mathfrak K}(E) \rightarrow \mathds{R}$ where $C_{\mathfrak K}(E)$ denotes the set of continuous functions with compact support, and $\phi$ positive means $\phi(f) \geq 0$ for any $f\geq 0$. A Borel measure is a measure defined on the $\sigma$-algebra 
generated by compact subsets of $E$ such that the measure of every compact subset is finite.
As the first result, we want to describe more precisely (than it was done in Introduction) the relation between a measure and a functional.
Let $E$ be a locally compact space, $\mu$ be a positive (Borel) measure on $E$ and $f$ be a continuous function on $E$ with a compact support. 
Consider the following map:

\begin{equation}
C_{\mathfrak K}(E) \ni f \mapsto \int fd\mu \in \Cn.\label{s1eq10}
\end{equation}
 It is an easy observation that

\begin{enumerate}
\item $\int (f+g)d\mu=\int fd\mu +\int gd\mu$
\item $\int cfd\mu=c\int fd\mu$ for  $c \in \Cn$
\item $\int fd\mu \geq 0$ for $f \geq 0$
\end{enumerate}

Above conditions define a positive linear form on the space $C_{\mathfrak K}(E)$. In other words, in (\ref{s1eq10}), each positive measure on $E$ defines a positive linear form on $C_{\mathfrak K}(E)$. 
The converse implication is given by the Markov-Riesz-Kakutani theorem, cf Section 69 in \cite{Ber}
or vol. I of \cite{choquet}.

\begin{thm}
\label{2.17}
If $\varphi$ is a linear, positive, continuous form on $C_{\mathfrak K}(E)$ then there exist a unique positive Borel measure $\mu$ on $E$ such that
\begin{equation}
\varphi(f)=\int_E fd\mu   \qquad  f \in C_{\mathfrak K}(E).
\end{equation}
If additionally $E$ is compact and $\varphi(\mathds{1})=1$ then $\mu$ is a probability measure. Consequently,
normalized forms and probability measures are in $1-1$ correspondence.
\end{thm}

To comment this result one may say that Theorem \ref{2.17} establishes a one-to-one correspondence between the positive Borel measures (in fact, regular positive Borel measures) and the Radon measures such that $\mu(E) = \| \mu\|$. The standard integral notation $\mu(f) = \int_E f d\mu$, $f\in C_{\mathfrak K}(E)$ implicitly identifies these concepts. Thus, in in sequel, \textit{the term measure will be used interchangeably to denote them}.

Assume additionally, temporary, that $E$ is a compact space. 
$C_{\mathfrak K}(E) \equiv C(E)$ equipped with the supremum norm $\|f\| = \sup_{x \in E}|f(x)|$, where $f \in C(E)$, is a Banach space and the map (\ref{s1eq10}) 
is continuous one (if $C(E)$ is considered as a Banach space).
One can say even more: $C(E)$ can be furnished with an algebraic structure turning $C(E)$ into an abelian $C^*$-algebra with unit (see Section 4 for definitions). Furthermore, calling normalized forms (\ref{s1eq10}) as states one gets a \textit{one-to-one correspondence between a probability measure on $E$ and the corresponding state}.

For a Borel measure $\mu$ on a locally compact Hausdorff space $E$ and a function $f \in C_{\mathfrak K}(E)$ we denote $\mu(f) = \int_E f d\mu$ (cf Theorem \ref{2.17}). Denote by $\mathfrak{M}(E)$ the collection of Radon measures on $E$.
Let $\{ \mu_n \}_{n=1}^{\infty} \subset
\mathfrak{M}(E)$.
We say that the net $\{ \mu_n \}$ is  weakly convergent to 
$\mu$ if $\mu_n(f)\rightarrow \mu(f)$ for any function $f \in C_{\mathfrak K}(E)$. This topology of simple convergence is called the vague topology (and sometimes also called the weak$^*$-topology).

We will need the notion of Dirac's (point) measure $\delta_a$, where $a \in E$. Such measures are determined by the condition:

\begin{equation}
\delta_a(f)=f(a)
\end{equation}

We say that a measure $\mu$ has a finite support if it can be written as a linear (finite) combination of $\delta_a$'s.
Now, we are in position to give, see Chapter 3, Section 2 , Corollaire 3 in \cite{Bour}:
\begin{thm}
\label{2.14}
Any positive finite measure $\mu$ on $E$ is a limit point, in the vague topology, of a convex hull of positive measures 
having a finite support contained in the support of $\mu$. \newline
\end{thm}

\begin{rem}
\label{2.15}
\begin{enumerate}
\item This result will be not valid in the non-commutative setting. It is taken from the (classical) measure theory.
\item A slightly stronger formulation can be find in \cite{Mey}.  Namely, every probability measure $\lambda$ in
$\mathfrak{M}(E)$ is a weak limit of discrete (with finite support) measures belonging to the collection of probability measures in 
$\mathfrak{M}(E)$ which have the same barycenter as $\lambda$ (cf Definition \ref{7.1}).
\item The statement of Theorem \ref{2.14} can be rephrased by saying that a classical measure has the weak-$^*$ Riemann approximation property.
\end{enumerate}
\end{rem}

The final result, given in this section, is a preparatory one for a description of (classical) composite systems.
Let ($\Omega_i$, $\mathcal{F}_i$, $p_i$), $i=1,2$ be a probability space. 
The product of two probability spaces (which is also a probability space)  
is the Cartesian product ($\Omega_1\times \Omega_2$, $\mathcal{F}_1\times \mathcal{F}_2$, $p_1 \times p_2$), 
where the product measure $p_1 \times p_2$ is defined as $p_1 \times p_2(A\times B)=p_1(A)p_2(B)$ for all 
$A\in \mathcal{F}_1$ and $B\in\mathcal{F}_2$. $\mathcal{F}_1\times \mathcal{F}_2$ denotes the $\sigma$-algebra generated by sets of the form $\{ A \times B; A \in \mathcal{F}_1, B \in \mathcal{F}_2 \}.$

\begin{rem}
Let $\chi_{_{Y}}$ be an indicator function of a measurable set $Y$.
Assume that for any $A\in \mathcal{F}_1$ and $B\in \mathcal{F}_2$ 
functions $\chi_{A\times \Omega_2}$ and $\chi_{\Omega_1\times B}$ are uncorrelated. 
Then one has 
\begin{equation}
\mu(\chi_{A\times \Omega_2}\cdot\ \chi_{\Omega_1\times B})=\mu(\chi_{A\times \Omega_2})\mu(\chi_{\Omega_1\times B}) \label{s1eq9}
\end{equation}
The measures $\mu_1(A)=\mu(\chi_{A\times \Omega_2})$, $\mu_2(B)=\mu(\chi_{\Omega_1\times B})$ are called the marginal measures. Consequently, if any two events give rise to uncorrelated indicator functions, for a measure $\mu$ 
on the product measurable (probability) space then $\mu=\mu_1\times\mu_2$.
\end{rem}

\section{Classical composite systems}
 
In classical statistical mechanics, a system is described by its phase space $\Gamma$, a probability measure $\mu$, and a one parameter family of measure preserving maps $T_t: \Gamma \rightarrow \Gamma$. The phase space $\Gamma$ records the allowed collection of system coordinates and momenta. The measure $\mu$ is characterizing our knowledge about the system.
In particular, the case when $\mu$ is a point measure corresponds to the full knowledge about the system. 
The important point to note here is that such form of a measure leads to classical mechanics.
Finally, the family of maps $\{T_t\}$, where either $t \in \Rn$  or $t \in \Rn^+$, is designed to describe a time evolution of the system. Suppose we have two such systems 
 ($\Gamma_i$, $\mu_i$, $T_t^i$), $i=1,2.$ We wish to form one bigger system consisting of these two given sub-systems
 - thus we want to form a \textit{composite} system. But, prior to any discussion on this matter one should realize that there are
 three different types of independence cf \cite{Streat}. Namely:
\begin{enumerate}
\item logical independence, 
\item dynamical independence, 
\item statistical independence. 
\end{enumerate}

Logical independence means that we are implicitly assuming that the values allowed for the first component (so for $\Gamma_1$)
do not depend on the values taken by the second component (so those in $\Gamma_2$). This leads to the conclusion that the allowed values for the composite system are given by the Cartesian product of its components, i.e. $\Gamma = \Gamma_1 \times \Gamma_2$. From now on we make this assumption.
The dynamical independence means that the global evolution is described by the product of maps, i.e. $T_t=T_t^1\times T_t^2$. The important point to note here is that this independence excludes any interaction between the two subsystems. Therefore, this kind of independence will not be assumed. The last one, the statistical independence means that the global measure describing our knowledge about the system is a product measure. It is a simple matter to check that this independence is related to that described in Section 2. Therefore, we will not assume the statistical independence.
To sum up: \textit{a composite system is characterized by the triple  $(\Gamma \equiv \Gamma_1 \times \Gamma_2$, $\mu$, $T_t)$, where the probability measure $\mu$ is defined on the Cartesian product of two measurable  spaces $(\Gamma_1 \times \Gamma_2, {\cal F}_1 \times {\cal F}_2)$, and finally, $T_t$ is a global evolution defined on $\Gamma$}.

Having clarified the idea of a composite system we want to study correlation functions. In the reminder of this section, for simplicity, we assume that $\Gamma_1$, $\Gamma_2$ are compact sets and we consider only continuous stochastic variables. Thus, a global (classical) observable is given by a function $f \in C(\Gamma_1 \times \Gamma_2)$ while an observable associated with a subsystem 
is given by $f_i \in C(\Gamma_i)$, $i= 1,2$ respectively.
We recall that there is the identification (see Section 22 in \cite{sakai} or Section 11.3 in \cite{KR})
\begin{equation}
\label{14}
C(\Gamma_1 \times \Gamma_2) = C(\Gamma_1)\otimes C(\Gamma_2)
\end{equation}
where on the right hand side of (\ref{14}) $\otimes$ stands for the tensor product, (for more details on tensor products we refer the reader to Section 5). The advantage of using $(\ref{14})$ lies in the fact that it is now easy to identify observables associated with a subsystem, the given form is easy for the quantization, and finally one can naturally embed a subsystem into a large composite system. 

Now we wish to study a two point correlation function, where an observable $f_i$ is associated with a subsystem $i$, $i=1,2$. We  note, that quantum counterparts of such correlation functions are typical in Quantum Information Theory. To simplify our notation we will identify the function $f_1$ (defined on $\Gamma_1$) with the function $f_1 \otimes \mathds{1}_{\Gamma_2}$ (defined on $\Gamma_1 \times \Gamma_2$); and analogously for $f_2$.
Thus, we wish to study the functionals $\varphi(\cdot)$ given by
\begin{equation}
\varphi(f_1 \otimes f_2) =
\varphi(f_1f_2) \equiv \varphi_{\mu}(f_1 f_2)\equiv \int_{\Gamma_1\times \Gamma_2}f_1(q_1)f_2(q_2)d\mu
\end{equation}
where $f_i \in C(\Gamma_i)$, $i=1,2$ and we have used Theorem \ref{2.17}.

Now taking into account the weak-$^*$ Riemann approximation property, see Theorem \ref{2.14}, one has 

\begin{equation}
\begin{split}
\varphi(f_1f_2)&=\lim_{n \to \infty} \int_{\Gamma_1\times \Gamma_2} f_1({\gamma}_1)f_2({\gamma}_2)d\mu_n\\
&=\lim_{n \to \infty} \int_{\Gamma_1\times \Gamma_2} f_1({\gamma}_1)f_2({\gamma}_2)(\sum_{n} \lambda_nd\delta_{(a_{1,n}, a_{2,n})}^n)      \label{s1eq11}
\end{split}
\end{equation}
where $\delta_{(a,b)}$ stands for the Dirac's measure supported by $(a, b)$, $\lambda_n \geq 0$ and $\sum_n \lambda_n = 1$. Note, that for a point measure, one has
\begin{equation}
\delta_{(a,b)}=\delta_a\times \delta_b.
\end{equation}
Therefore
\begin{equation}
\begin{split}
\varphi(f_1f_2)&=\lim_{n \to \infty}\sum_{n} \lambda_n  \int_{\Gamma_1} f_1({\gamma}_1)d\delta_{a_{1,n}}^n(\gamma_1)\int_{\Gamma_2}f_2({\gamma}_2)d\delta_{a_{2,n}}^n(\gamma_2)\\
&=\lim_{n \to \infty}\sum_{n} \lambda_n\varphi_{{\delta}_{a_{1,n}}}(f_1)\varphi_{{\delta}_{a_{2,n}}}(f_2)\\
&=\lim_{n \to \infty}\sum_{n} \lambda_n(\varphi_{{\delta}_{a_{1,n}}}\otimes \varphi_{{\delta_{a_{2,n}}}})(f_1\otimes f_2)    
    \end{split}             
\end{equation}

for any $f_i \in C(\Gamma_i)$, $i=1,2$.
Consequently
\begin{equation}
\varphi_{\mu}(f_1 \otimes f_2) = \lim_{n \to \infty}\sum_{n} \lambda_n(\varphi_{{\delta}_{a_{1,n}}}\otimes \varphi_{{\delta_{a_{2,n}}}})(f_1\otimes f_2)
\end{equation}
for any $f_i \in C(\Gamma_i)$, $i=1,2$.
\begin{cor}
\label{3.1}
For a classical case, any two point correlation function of bipartite system is the limit of
a convex combination of product states. 
\end{cor}

This means that it is of the separable form (for the definition see Section 6). We will see that  it is not true for 
the quantum case. The important point to note here is that two point correlation function is the principal ingredient of correlation coefficient.

\section{Rules of Quantum Mechanics}
Let us begin this section with a statement borrowed from the Omn\'es book, see page 111 in \cite{omnes}.
"\textit{Every physical system, whatever an atom or a star is assumed to be described by a universal 
kind of mechanic which is Quantum Mechanics}."\newline

Thus, it is obvious that one should quantize both probability calculus and the concept of (classical) composite systems which were given in Sections 2 and 3 respectively. To proceed to quantization, for the sake of clarity, it is desirable
to list the rules of quantum theory. To this end we extract them from Dirac's \cite{dirac} and Omn\'es \cite{omnes} books. One has:

\begin{rle} 
The theory of given individual isolated physical system can be entirely expressed in terms 
of a specific Hilbert space and mathematical notions associated with it, particularly a \textit{specific algebra of operators}.
\end{rle}

\begin{rle}
A specific self-adjoint operator is associated with an \textbf{isolated} physical system. This operator is the system Hamiltonian $H$. It determines Heisenberg's dynamics (which replace Newton's law of motion). The dynamics is expressed by a continuous $1$-parameter unitary group of operators {$U(t)$, $t\in \mathbb{R}$} having the Hamiltonian $H$ as its infinitesimal generator. In particular, the evolution of an observable $O$ is given by
$O_t=U(t)OU^*(t)\equiv \alpha_t(O)$ (in Heisenberg picture)
while the evolution of a state is given by
$\Psi_t=U(t)\Psi$ (in Schr\"odinger picture).
\end{rle}

\begin{rle}
On the specific algebra $\mathfrak{A}$ of operators associated with a given physical system there exist a family of linear positive normalized forms on $\mathfrak{A}$. They form the set of states $\mathfrak{S}$. The interpretation of a state $\varphi$ ($\varphi \in \mathfrak{S}$, $\mathfrak{A}\ni O\rightarrow \varphi(O) \in  \mathbb{C}$) as given by Born is that the number $\varphi(O)$ (real if $\varphi$ and $O$ are self-adjoint) is the expectation value of observable $O$.\newline
\end{rle}\noindent
and 

\begin{rle}
Let two physical systems $S_1$ and $S_2$ be represented by Hilbert spaces $\mathcal{H}_1$, $\mathcal{H}_2$, algebras $\mathfrak{A}_1$, $\mathfrak{A}_2$, sets of states 
$\mathfrak{S}_1$, $\mathfrak{S}_2$ and finally Hamiltonians $H_1$, $H_2$ respectively. When they are combined in one (composite) system $S$ then its Hilbert space $\mathcal{H}$ is equal  $\mathcal{H}_1\otimes \mathcal{H}_2$ and its algebra $\mathfrak{A}=\mathfrak{A}_1\otimes \mathfrak{A}_2$. When the systems $S_1$ and $S_2$ are dynamically independent,  
Hamiltonian associated with the composite system is given by $H=H_1\otimes \mathds{1}+\mathds{1}\otimes H_2$.
\end{rle}\noindent
and finally 

\begin{rle}
Let a composite system $(\mathfrak{A} \equiv \mathfrak{A}_1 \otimes \mathfrak{A}_2,\mathfrak{S}_{\mathfrak{A}}, \alpha_t(\cdot))$ be given. When, one is interested only in time evolution of a subsystem, say that labeled by $"1"$,
then a reduction of (global) Hamiltonian type dynamics should be carried out.
As a result, time evolution of the subsystem $"1"$ is described by a one parameter family of (linear) maps
$T_t : \mathfrak{A}_1 \rightarrow \mathfrak{A}_1$ such that $T_t(f) \geq 0$ for any $t$ and a positive $f \in \mathfrak{A}_1$ (so positivity preserving), and $T_t(\mathds{1}) = \mathds{1}$.
\end{rle}\noindent
To comment the last rule we add

\begin{rem}
\begin{itemize}
\item If one adds Markovianity assumption then $\{T_t\}$ is also a semigroup, i.e. $T_t \circ T_s = T_{t+s}$
for non-negative $t$ and $s$.
\item Frequently, more stronger assumption on positivity - \textit{complete positivity} -\textbf{} is relevant, see \cite{AlLe}, \cite{davies}.
\end{itemize}
\end{rem}\noindent
To illustrate these rules we present:

\begin{exmp} (\textit{Dirac's formalism})
\label{4.5}
\begin{enumerate}
\item (\textit{Dirac's Quantum Mechanics.})  A separable, infinite dimensional  Hilbert space $\mathcal H$ is associated with a system.
The specific algebra $\mathfrak A$ is taken to be  $B(\mathcal{H})$ - the set of all linear, bounded operators on $\mathcal H$. 
The set of states is determined by density matrices, i.e. positive trace class operators 
with trace equal to 1. The expectation value $<A>$ of $A=A^* \in B(\mathcal{H})$ at a state $\varrho$  is given by $<A> = \Tr \varrho A$
where $\varrho$ is a density matrix.
\item (\textit{composite system in Dirac's formalism})
Let $S_i$ be a physical system, $\mathcal{H}_i$ be an associated Hilbert space, $B(\mathcal{H}_i)$ 
be an associated algebra and $\mathfrak{S}_i$ denote set of states (density operators) for $i=1,2$. 
Then Hilbert space, associated algebra and set of states for composite system are given by $\mathcal{H}=\mathcal{H}_1\otimes \mathcal{H}_2$, 
  $B(\mathcal{H}_1\otimes \mathcal{H}_2)\approx B(\mathcal{H}_1)\otimes B(\mathcal{H}_2)$ and $\mathfrak{S}$ where, in general contrary to the classical case, one has $\mathfrak{S}\neq \overline{conv}( \mathfrak{S}_1\otimes \mathfrak{S}_2)$ (for more details, see the next section).
\end{enumerate}
\end{exmp}\noindent
and

\begin{exmp}(\textit{Classical systems})
\begin{enumerate}
\item (\textit{classical system})
Assume that $\mathfrak A$ is an abelian C*-algebra with unit $\mathds{1}$. Then (according to Gelfand-Neimark theorem 
(see eg. section 1.2 in \cite{sakai}) $\mathfrak A$ can be identified with the $C^*$-algebra (see below for the definition of $C^*$-algebra)
of all complex valued continuous functions on $\Gamma$, where $\Gamma$ is a compact Hausdorff space.
Hence, a state (normalized, positive, linear form) on $\mathfrak A$ leads to a probability measure on $\Gamma$, cf. Theorem \ref{2.17}. Consequently, the probabilistic scheme described in Section 2 was recovered.
\item (\textit{classical composite system})
Take (for $i=1,2$) 
 two abelian $C^*$-algebras $\mathfrak{A}_i$ with unit and combine them in one system (cf Rule 4.4).
 Then we have $\mathfrak{ A} \equiv \mathfrak{A}_1\otimes \mathfrak{A}_2=C(\Gamma_1)\otimes C(\Gamma_2)\approx C(\Gamma_1\times \Gamma_2)$. 
 Take a state (normalize, linear, positive form) on $\mathfrak A$. Then, according to Theorem \ref{2.17}, there exists a probability measure on $\Gamma_1 \times \Gamma_2$. 
Thus, we recover a classical composite system described in the previous section. In particular, (see the end of the previous section; Corollary \ref{3.1}) one has $\mathfrak{S}= \overline{conv}(\mathfrak{S}_1\otimes \mathfrak{S}_2)$!
\end{enumerate}
\end{exmp}

To comment the just presented Rules we need some definitions. 
\begin{defn}
Let $\mathfrak A$ be a Banach space. We say that $\mathfrak A$ is a
Banach algebra if it is an algebra, i.e.
a multiplication
\begin{equation}
\mathfrak{A} \times \mathfrak{A} \ni \left\langle f,g\right\rangle \mapsto fg \in \mathfrak{A}
\end{equation}
is defined in such way
that for every $f, g\in \mathfrak A$ 
one has
\begin{equation}
\left\|fg\right\|\leq\left\|f\right\|\left\|g\right\|
\end{equation}
It follows that multiplication in Banach algebra is separately continuous in each variable.
\end{defn}

\begin{defn}
An involution on an algebra $\mathfrak A$ is a antilinear map $f\rightarrow f^*$ such that for all 
 $f, g\in \mathfrak A$  and $\lambda\in \mathbb{C}$ one has 
\begin{equation}
(f^*)^*=f
\end{equation}
\begin{equation}
(fg)^*=g^*f^*
\end{equation}
\begin{equation}
(\lambda f)^*=\overline{\lambda}f^* \quad {\rm{and}} 
\quad (f+g)^* = f^* + g^*
\end{equation}
A $*$-algebra is an algebra with involution. A $*$-Banach algebra $\mathfrak A$ is a $*$-algebra such that $\mathfrak A$ is a Banach algebra and $\|f^*\| = \|f\|$.
\end{defn}

\begin{defn}
A C*-algebra $\mathfrak A$ is a $*$-Banach algebra such that the norm $\| \cdot \|$ satisfies
\begin{equation}
\left\|ff^*\right\|=\left\|f\right\|^2
\end{equation}
\end{defn}

Finally, cf Section 5.1 in \cite{KR}:

\begin{defn}
A $C^*$-algebra $\mathfrak M$, acting on a Hilbert space $\mathcal H$, that is closed in the weak operator topology and contains the unit $\mathds{1}$ 
is said to be a von Neumann algebra (or equivalently, a $W^*$-algebra).
\end{defn}

We will need also a warning.
 Namely, see \cite{Win}, \cite{Wiel}, \cite{Sakai1} 
\begin{prop}
\label{4.11}
It is impossible to find two elements  $a$, $b$ in a Banach algebra $\mathfrak A$ such that 
\begin{equation}
ab - ba = 1.
\end{equation}
\end{prop}

The principal significance of Proposition \ref{4.11} stems from the following conclusion:  \textit{it is impossible to realize canonical commutation relations in terms of a Banach algebra}. Thus, in particular,
it is impossible to carry out a canonical quantization on finite dimensional spaces.
Consequently, \textbf{ it is difficult to speak about quantumness of finite dimensional systems}.

Now we are in a position to comment the above listed Rules for quantization. The specific algebra mentioned in Rule 4.1
means the collection of bounded functions of, in general, unbounded operators. Note that such operators appear in the theory due to the procedure of quantization. On the other hand, Proposition \ref{4.11} clearly shows that finite dimensional models are not able to describe genuine quantum systems. Therefore, they can only provide so called ``toy'' models. Such models are frequently useful as they can shed some new light on very specific questions which are appearing in non-commutative setting.

Rule 4.3 is saying that a (quantum) observable is a non-commutative counterpart of a stochastic variable, cf Remark \ref{2.8}.
In particular, this implies that quantum probability should appear. The standard form of non-commutative  probability calculus is provided by the pair $(\mathfrak{A}, \varphi)$, where $\mathfrak A$ is a $C^*$-algebra, $\varphi$ is a state, i.e. a linear functional on $\mathfrak A$ such that $\varphi(\mathds{1}) = 1$, and $\varphi(a) \geq 0$ for all $a \geq 0$. It is important to note here that on the additional assumption that $\mathfrak A$ is abelian one can recover the classical probability scheme. Namely, any abelian $C^*$-algebra $\mathfrak A$ with unit $\mathds{1}$ can be identified with the collection of all complex valued functions defined on a compact Hausdorff space $E$.
Therefore, Theorem \ref{2.17} implies the existence of a probability measure $\mu$ on $X$, which is uniquely determined by a state. Consequently, fundamentals of a (classical) probability were obtained. But, as it is impossible to embed a non-commutative $C^*$-algebra into commutative one, there is no hope to embed quantum probability into larger classical probability scheme.

There is also another extremely important motivation for more ``refined'' algebras - for $W^*$-algebras.
Namely, classical mechanics demands the differential and integral calculus for its description. It is natural to expect that quantum mechanics demands non-commutative calculus for its description. 
This is the case. In fact, it was recognized in early days of quantum mechanics, see for example von Neumann book \cite{Neumann}. The non-commutative calculus, including non-commutative integration theory, was developed in the second half of XX century, see \cite{Ne}, \cite{Te1}, and \cite{takesaki}, and it is based on $W^*$-algebra theory. For example, such calculus was successfully applied for 
quantum potential theory, see \cite{Franz}, for
a quantization of Markov-Feller processes, see \cite{R7}, \cite{R8}, \cite{R9} and for statistical mechanics, see \cite{LM}.

Consequently, to simplify our exposition of quantum rules, by a \textit{specific algebra} we will mean either $C^*$-algebra or $W^*$-algebra.

The Rule 4.4 is saying that to form a bigger system consisting of two smaller subsystems, a tensor product of appropriate algebras should be taken. As this concept for Banach spaces is not trivial one, we will clarify this notion in the next section.

We wish to close this section with the final remark concerning the following question:
Why we do not restrict ourselves to Dirac's formalism only? In other words, why we will not restrict ourselves to
formalism specified in Example \ref{4.5}? The answer follows from the following observations, for details see \cite{haag} and \cite{BR}:
\begin{enumerate}
\item Individual features of a system as well as a relation between a system and the region occupied by this system should be taken into account, (see also Examples \ref{7.15}, \ref{7.13}, and \ref{7.19})
\item To take into account a causality one should have a possibility to say how far away is a subsystem $S_1$ 
from $S_2$ (in classical case, subsystems are described by subspaces of an
Euclidean space. So this question has an easy solution). The above question, for quantum case, demands more general setting than that one which is offered by Example \ref{4.5}.
\item Quantum field theory demands more general approach than that given by Dirac's formalism, see \cite{haag}. The important point to note here is that quantum correlations as a phenomenon was observed in quantum field theory many decades ago. Probably, the best example is given by the Reeh and Schlieder theorem \cite{Reeh}. This theorem is saying that any state vector of a quantum field can be approximated by an action of ``local'' operator acting on the vacuum. To this end the operator must exploit the small but non zero long distant correlations which exist in the vacuum.
Other examples of quantum correlations in quantum field theory can be found in \cite{SW1} and \cite{SW2}.
\end{enumerate}

\section{Tensor products}
This section provides a brief exposition of the theory of tensor products, for details we refer the reader to \cite{Ryan}, \cite{DF}, and \cite{DFS}. We begin with the definition of tensor product of linear spaces.
\begin{defn} 
\label{5.1}
Let $X$ and $Y$ are linear spaces. We say that a form $\varphi:X\times Y\rightarrow \mathbb{C}$ 
is bilinear if $\varphi(\lambda_1x_1+\lambda_2x_2,\lambda_3y_1+\lambda_4y_2)=\lambda_1 \lambda_3\varphi(x_1,y_1)+ \lambda_1 \lambda_4\varphi(x_1,y_2)+\lambda_2 \lambda_3\varphi(x_2,y_1)+ \lambda_2 \lambda_4\varphi(x_2,y_2)$ where $x_1,x_2 \in X$, $y_1, y_2 \in Y$ and $\lambda_i \in \mathbb{C}$, $i=1,...,4$.  
Let us denote the linear space of all such bilinear forms by $B(X, Y)$. 
A simple tensor is a linear form $x\otimes y$ on the linear space $B(X, Y)$, ie. $x\otimes y \in B(X, Y)^{\prime}$, such that:

\begin{equation}
(x\otimes y)(A)=A(x,y)       \label{s1eq13}
\end{equation}
for all $A\in B(X, Y)$.
The algebraic tensor product of two linear spaces $X$ and $Y$, $X\odot Y\subset B(X, Y)^{\prime}$, is the set of all linear, finite, 
combinations of simple tensors. A typical $v\in X\odot Y$ is of the form:

\begin{equation}
v=\sum_{i=1}^n\lambda_i x_i\otimes y_i    \label{s1eq14}
\end{equation}
where we emphasize that the decomposition given by (\ref{s1eq14}) is not unique.
\end{defn}

The above construction can be applied to Banach spaces $X$ and $Y$. However, 
the algebraic tensor product of $X$, $Y$ is not automatically a Banach space. 
To obtain tensor product of Banach spaces which itself is a Banach space one 
must define a norm on $X\odot Y$. However, \textit{unlike the finite dimensional case, the tensor product of infinite dimensional Banach spaces behaves mysteriously}, see \cite{takesaki}, vol I. Notes to Section IV.5. In particular,
on $X\odot Y$ one can define various norms. 
It is natural to restrict oneself to norms
satisfying: $\left\|x \otimes y \right\|=\left\|x\right\|\left\|y\right\|$. Such norms 
are said to be cross-norms.
There are exceptional cases, where the cross-norm is uniquely defined. The most important case for Physics 
is that one when Banach spaces $X$ and $Y$ are Hilbert spaces and we want $X \otimes Y$ to be also a Hilbert space. However, in general,  there are plenty of cross-norms on $X\odot Y$. To illustrate this 
phenomenon we give:

\begin{exmp}\textit{Operator norm.}\newline
Let $\mathcal{H}_i$ be a Hilbert space and $B(\mathcal{H}_i)$ denote the space of all 
bounded linear operators on $\mathcal{H}_i$ (for $i=1,2$). Then the operator norm on 
$B(\mathcal{H}_1)\odot B(\mathcal{H}_2)\subseteq B(\mathcal{H}_1\otimes \mathcal{H}_2)$ 
is taken from that on $B(\mathcal{H}_1\otimes \mathcal{H}_2)$. It has the cross-norm property.
The closure of $B(\mathcal{H}_1)\odot B(\mathcal{H}_2)$ with respect to this operator norm will be denoted by $B(\mathcal{H}_1)\otimes B(\mathcal{H}_2)$ and called the tensor product of $B(\mathcal{H}_1)$ and 
$B(\mathcal{H}_2)$.
\end{exmp}

We now define two important others cross-norms on $X \odot Y$:

\begin{exmp}\textit{Injective norm.}\newline
Let $X$ and $Y$ be Banach spaces. The injective norm $\epsilon$ on  $X\odot Y$ is defined  by
\begin{equation}
\epsilon(v)=\sup\left\{|\sum_{i=1}^n f(x_i) g(y_i)|: f \in X^*, \|f\| \leq1; g \in Y^*, \|g\| \leq 1 \right\}
\end{equation}
where $v=\sum_{i=1}^nx_i\otimes y_i$. The completion of
$X\odot Y$ with respect to the norm $\epsilon$ is called the injective tensor product and is denoted by $X\otimes_\epsilon Y$.
\end{exmp}\noindent
and
\begin{exmp}\textit{Projective norm.}\newline
Let $X$ and $Y$ be Banach spaces. The projective norm $\pi$ on  $X\odot Y$ is defined  by
\begin{equation}
\pi(v)=\inf\left\{\sum_{i=1}^n\left\|x_i\right\|\left\|y_i\right\|:\ v=\sum_{i=1}^nx_i\otimes y_i\right\}
\end{equation}
The completion of
$X\odot Y$ with respect to the norm $\pi$ is called the projective tensor product and is denoted as $X\otimes_\pi Y$.
\end{exmp}

Both norms, injective and projective, have the cross-norm property. Moreover, the projective norm is the largest  cross-norm while the injective norm is the smallest cross-norm. Moreover, if $\mathfrak{A}_1$ and $\mathfrak{A}_2$ are $C^*$-algebras, $\| \cdot \|_{\pi}$ is rarely a $C^*$-norm. 
The importance of projective norms follows from the old Grothendieck result \cite{Gro}, see also \cite{DF}.

\begin{thm}
\label{5.5}
Let $X$ and $Y$ be Banach spaces.
Then, there exists an isometric isomorphism between the Banach space $\mathfrak{B}(X,Y)$
of all bounded bilinear functionals on $X\times Y$ and the space $(X\otimes_{\pi} Y)^*$ of all continuous linear functionals on $(X\otimes_{\pi} Y)$ given by 
\begin{equation}
\hat{\varphi}(x\otimes y)=\varphi(x,y)
\end{equation} 
where $\varphi \in \mathfrak{B}(X,Y)$, $x \in X$, and $y \in Y$.
\end{thm}\noindent
The modification of the projective norm, the operator space projective norm, leads to, see \cite{EfRu}, Section 7.2,
\begin{thm}
\label{5.6}
Let $\mathfrak{M}\subseteq B(\mathcal{H})$ and $\mathfrak{N}\subseteq B(\mathcal{K})$ be two von Neumann algebras.
Denote by $\mathfrak{M}_*$ the predual of $\mathfrak M$, i.e. such Banach space that $(\mathfrak{M}_*)^*$ is isomorphic to $\mathfrak M$, i.e. $(\mathfrak{M}_*)^*\cong \mathfrak M$.
There is an isometry 
\begin{equation}
(\mathfrak{M} \otimes \mathfrak{N})_* = \mathfrak{M}_* \hat{\otimes}_{\pi} \mathfrak{N}_*
\end{equation}
where the von Neumann algebra $\mathfrak{M} \otimes \mathfrak{N}$ is the weak closure of the set $\{ A \otimes B; A \in \mathfrak{M}, B\in \mathfrak{N} \}$.
In particular, 
\begin{equation}
B(\mathcal{H} \otimes \mathcal{K})_* = B(\mathcal{H})_* \hat{\otimes}_{\pi} B(\mathcal{K})_*.
\end{equation}
 $\mathfrak{M}_* \hat{\otimes}_{\pi} \mathfrak{N}_*$ denotes the closure of $\mathfrak{M}_* \otimes \mathfrak{N}_*$
with respect to the operator space projective norm.
\end{thm}
The above given identifications \textit{will be used for definition of entangled states} on von Neumann algebras, see the next section. However, here, we wish to make two remarks.
\begin{rem}
As it was mentioned, the identification given by Theorem \ref{5.6} will be used to define entangled normal states. On the other hand, using the theory of tensor products of (abstract) $C^*$ and $W^*$-algebras, see sections IV.4 and IV.5 in \cite{takesaki}, one can show that there is an equivalent way to describe the predual of the tensor product of $W^*$-algebras. However, as the emphasis will be put on the $C^*$-algebraic case, neither the detailed description of the predual of the tensor product of $W^*$-algebras nor the explicit form of the norm will be used. This explains why we drop further details on the operator space projective norm.
\end{rem}
and
\begin{rem}
One can show, see Chapter 1 in \cite{sakai}, or \cite{KR}, that $\mathfrak{M}_*\subset \mathfrak{ M}^*$.
Consequently, any element $\varphi$ in $\mathfrak{M}_*$  is a linear bounded form on $\mathfrak M$. But, being in $\mathfrak{M}_*$, $\varphi$ possesses extra continuity properties. It is weak-$^*$ continuous or equivalently $\varphi$ is called to be normal. The ``physical'' significance of such states follows from the fact that normal states correspond to \textit{density matrices}.
\end{rem}

The important property of the injective norm is contained in (see \cite{Gro} and 
\cite{DF}):
\begin{thm}
\label{5.8}
Let $\varphi$ be a linear form on $E\odot F$ where $E$ and $F$ are Banach spaces. Then, $\varphi$ is continuous with respect to the injective norm,\newline
 $\varphi\in (E\otimes_{\epsilon} F)^*$,
if and only if there exist (positive) 
Borel measure $\mu$ on $B_{E^*}\times B_{F^*}$ (so on the Cartesian product of unit balls in duals of $E$ and 
$F$ with the weak $*$-topology) such that for all $z\in E\odot F$

\begin{equation}
\label{integral}
\left\langle \varphi,\ z \right\rangle=\int_{B_{E^*}\times B_{F^*}}\left\langle x^*\otimes y^*,\ z \right\rangle d\mu(x^*,\ y^*)
\end{equation}
The measure $\mu$ can be chosen such that
\begin{equation}
\left\|\mu\right\|:=\mu(B_{E^*}\times B_{F^*})=\left\|\varphi\right\|_{(E\otimes_{\epsilon} F)^*}
\end{equation}
\end{thm}

Formula (\ref{integral}) explains why bilinear forms in $(E\otimes_{\epsilon} F)^*\subset (E\otimes_{\pi} F)^* = \mathfrak{B}(E,F)$
will be called integral. It is worth pointing out that the measure $\mu$ can be taken to be positive. Consequently, there is a similarity between the form of integral  forms and separable states, cf Corollary \ref{3.1}. However, this similarity is somehow misleading as 
in (\ref{integral}) any order relation (positivity) on $E$ and $F$ is not taken into account.

\begin{rem}
To see, in details, the above similarity we need some preliminaries. Assume additionally that $E$ ($F$) are ordered Banach spaces, i.e. $E$ ($F$) contains a cone $E^+$ ($F^+$ respectively). Denote by $B^+_{E^*}$ ($B^+_{F^*}$) the set of all positive functionals $\varphi$ ($\psi$) on $E$ ($F$ respectively) such that $\varphi(f) \geq 0$ for $f \in E^+$ and $\|\varphi \| \leq 1$ (and analogously for $B^+_{F^*}$). It is easy to observe that if the measure $\mu$ employed in Theorem
\ref{5.8} is supported on $B^+_{E^*} \times B^+_{F^*}$ then (\ref{integral}) provides states of the same form as those given by Corollary \ref{3.1}. In other words, there is a possibility to get a family of ``separable'' states within the very general framework of Grothendieck approach to the theory of tensor products.
\end{rem}

The projective tensor product gains in interest if we realize that Rules $4.1-4.3$ and $4.5$ provide a nice example
of application of this tensor product.
Namely,
let $\mathfrak{A}$ stands for algebra of observables. We assume $\mathfrak A$ is either a $C^*$-algebra or (when speaking about normal states) a $W^*$-algebra. Obviously, always, it is a Banach space.
The set of all states $\mathfrak{S}$ is a subset of  $\mathfrak{A}^*$ (and $\mathfrak{A}^*$ is also a Banach space) while the collection of all density matrices gives a subset of $\mathfrak{A}_*$ ( $\mathfrak{A}_*$ is the predual of $\mathfrak{A}$, so it is also a Banach space).
The Born interpretation, cf Rule 4.3, implies
\begin{equation}
\label{32}
\mathfrak{A}\times \mathfrak{S} \ni \left\langle A, \varphi\right\rangle\rightarrow \varphi(A)\in \mathbb{C}
\end{equation}\noindent
where $\varphi(A)$ is interpreted as the expectation value of $A$ at the state $\varphi \in \mathfrak S $. 
Thus the Born's interpretation of Quantum Mechanics gives an element of $\mathfrak{B}(\mathfrak{A},\mathfrak{A}^*)$
since the  form $\hat{A}(\cdot, \cdot)$ on $\mathfrak{A} \times \mathfrak{S}$ defined by (\ref{32}) can be extended to the bilinear continuous form on $\mathfrak{A} \times \mathfrak{A}^*$ (or on $\mathfrak{A} \times \mathfrak{A}_*$, if one was interested in density matrices only).

Clearly, (\ref{32}) provides only one specific form on $\mathfrak{A} \times \mathfrak{A}^*$ (or on $\mathfrak{A} \times \mathfrak{A}_*$ respectively). However, it is crucial to note that Rule 4.3 combined with Rule 4.5 leads to the following recipe
\begin{equation}
\label{33}
\mathfrak{A}\times \mathfrak{S} \ni \left\langle A, \varphi\right\rangle\rightarrow \varphi(T_t(A))\in \mathbb{C}
\end{equation}
where $T_t \in \{T_t\}$ is a dynamical map. Obviously, in this way we are getting the large collection of bilinear, continuous forms.
 
 On the other hand, Theorem \ref{5.5} says
\begin{equation}
\mathfrak{B}(\mathfrak{A},\mathfrak{A}^*)\cong (\mathfrak{A}\otimes_{\pi}\mathfrak{A}^*)^* \label{grot}
\end{equation}
If the set of states $\mathfrak S$ consists of normal states only (so, for example, the collection of density matrices in Dirac's formalism of quantum mechanics is relevant) 
then one can rewrite (\ref{grot}) as
\begin{equation}
\mathfrak{B}(\mathfrak{A},\mathfrak{A}_*)\cong (\mathfrak{A}\otimes_{\pi}\mathfrak{A}_*)^*
\end{equation}
Then, using the another identification (again due to Grothendieck results, see \cite{Gro})
\begin{equation}
\label{40}
L(\mathfrak{A},\mathfrak{A})\cong (\mathfrak{A}\otimes_{\pi}\mathfrak{A}_*)^*,
\end{equation}
where $L(\mathfrak{A},\mathfrak{A})$ stands for the set of all bounded linear maps from $\mathfrak{A}$ to $\mathfrak{A}$,
it is not difficult to see that 
\begin{equation}
\label{41}
\mathfrak{B}(\mathfrak{A},\mathfrak{A}_*)\cong L(\mathfrak{A},\mathfrak{A})
\end{equation}

Therefore, (\ref{33}) gives not a large collection of continuous forms - \textbf{It gives the whole set of bilinear continuous forms} on $\mathfrak{A} \times \mathfrak{A}_*$, which are also positive, i.e. $\varphi(T_t(A)) \geq 0$ 
if $A\geq 0$ and $\varphi \geq 0$. In other words, \textit{Grothendieck theory of tensor products is perfectly compatible with the quantization rules of quantum mechanics}. One can say even more. Namely, the Grothendieck approach distinguishes also the scheme based on density matrices from that which is based on general (non-normal) states, for more details see the next section.

Here, we restrict ourselves to the remark that in particular, a subset of the family of positive maps 
$L^+(\mathfrak{A},\mathfrak{A}) \equiv \left\{T:\mathfrak{A}\rightarrow\mathfrak{A};\rm{linear}\ \rm{and}\ T(\mathfrak{A}^+)\subset\mathfrak{A}^+\right\}$ will play an extremely important role in an analysis of specific subsets of states, see Section 10.

We wish to end this section with a brief summary on the theory of tensor product of C*-algebras. We need these results for an implementation of Rule 4.4 which gives a starting point for a description of quantum composite systems.
Let $\mathfrak{A}_1$ and $\mathfrak{A}_2$ be C*-algebras with unit. Obviously,
 $\mathfrak{A}_1 \odot\mathfrak{A}_2$ can be constructed as before since $\mathfrak{A}_i$, $i=1,2$,  is also a Banach space).
As we wish to get a tensor product which is still a C*-algebra, we must define a C*-norm $\alpha$ on 
 $\mathfrak{A}_1 \odot\mathfrak{A}_2$ i.e. a norm such that $\alpha(x^*x)=(\alpha(x))^2$ and $\alpha(xy)\leq \alpha(x)\alpha(y)$.
Again, in general, see \cite{takesaki} there are plenty of such norms. As usually we will consider concrete $C^*$-algebras. Thus,  we will  use the operator norm. Consequently, the completion of the algebraic tensor product $\mathfrak{A}_1 \odot\mathfrak{A}_2$
with respect to the operator norm will be denoted by $\mathfrak{A}_1 \otimes\mathfrak{A}_2$.

In some cases (for so called nuclear C*-algebras) the tensor product is uniquely defined.
Nice examples of such algebras are provided by abelian algebras $\mathcal A$ (again, classical systems offer a great simplification) as well as  $M_n(\mathbb{C})$, where $n< \infty$. We remind that $M_n(\mathbb{C})$ can be considered as an approximation of physical situation, (cf the discussion following Proposition \ref{4.11}). Moreover,  they are the basic tool in many papers in Quantum Information Theory.
However, as it was mentioned, genuine quantum systems need infinite dimensional spaces. But, this means that tensor products used in the theory of composite systems as well as their properties depend on the proper choice of the norm.  

\section{Quantum composite systems}
In this section we want to describe the quantization of classical composite systems (cf Section 3). Analogously to description of classical composite systems we will assume only the logical independence. Thus, neither dynamical nor statistical independence will not be assumed. Therefore, a quantum composite system will be determined by the quadruple
\begin{equation}
\left(\mathfrak{A}\equiv \mathfrak{A}_1\otimes \mathfrak{A}_2, \mathfrak{S} \equiv \mathfrak {S}_{\mathfrak{A}}, \{T_t\}, \varphi_0\right)
\end{equation}
where $\mathfrak{A}$ (so also $\mathfrak{A}_i$) stands either for a $C^*$-algebra or a $W^*$-algebra. If $\mathfrak A$ is a $C^*$-algebra ($W^*$-algebra) then $\mathfrak S$ stands for the set of all states on the global system $\mathfrak A$ (all density matrices - normal states on global system respectively). $\{T_t\}$ stands for the set of dynamical maps while $\varphi_0$ is a distinguished state (playing the role of distinguished probability measure). Frequently, it is convenient to think that the algebra $\mathfrak{A}_i$ is associated with some particular region (in $\mathbb{R}^k$), $i=1,2$, cf examples given in the next section.

Let us consider the form of the second ingredient $\mathfrak S$ of the above definition (so the set of normalized, positive linear forms on $\mathfrak{A}_1\otimes \mathfrak{A}_2$).
The first attempt, following the classical case, would be to put  $\mathfrak{S}=\mathfrak{S}_1\otimes \mathfrak{S}_2$ or $\mathfrak{S}=\overline{conv}(\mathfrak{S}_1\otimes \mathfrak{S}_2)$.
Surprisingly these sets do not contain all states. Namely, one has (see Exercise 11.5.11 in \cite{KR})
\begin{exmp}
\label{exm6.1}
Let
$\mathfrak{A}_1=B(\mathcal{H})$ and $\mathfrak{A}_2=B(\mathcal{K})$ where $\mathcal{H}$ and $\mathcal{K}$ are 2-dimensional Hilbert spaces.
Consider the vector state $\omega_x(\cdot)=(x,\cdot\ x)$ with $x=\frac{1}{\sqrt{2}}(e_1\otimes f_1+e_2\otimes f_2)$ where $\left\{e_1,e_2\right\}$ and 
$\left\{f_1,f_2\right\}$  are orthonormal bases in $\mathcal{H}$ and $\mathcal{K}$ respectively.
Let $\rho$ be any state in the norm closure of the convex hull of product states, i.e. $\rho\in \overline{conv}(\mathfrak{S}_1\otimes \mathfrak{S}_2)$. Then, one can show that 
\begin{equation}
\left\|\omega_x-\rho\right\|\geq\frac{1}{4}.
\end{equation}
\end{exmp}

\begin{rem}
The reader should note that $\omega_x$ can always be approximated by a finite linear combination of simple tensors (cf Definition \ref{5.1} and formula (\ref{s1eq14})). However, here we wish to approximate $\omega_x$ by a convex combination of positive (normalized) functionals (cf Theorem \ref{2.14}) and this makes the difference.
\end{rem}

Consequently, contrary to the classical case (see Corollary \ref{3.1}) even in the simplest non-commutative case, the space of all states of $\mathfrak{A}_1\otimes \mathfrak{A}_2$
is not norm closure of $conv(\mathfrak{S}_1\otimes \mathfrak{S}_2)$. \textit{It means, in mathematical terms,
that for non-commutative case the weak$^*$ Riemann approximation property of a (classical) measure does not hold} (cf Remark \ref{2.15}(3)). Thus, we are in position to give:
\begin{defn}
\label{6.1}
\begin{itemize} 
\item \textit{$C^*$-algebra case.}\newline
Let $\mathfrak{A}_i$, $i=1,2$ be a $C^*$-algebra, $\mathfrak S$ the set of all states on $\mathfrak A\equiv \mathfrak{A}_1 \otimes \mathfrak{A}_1$, i.e. the set of all normalized positive forms on $\mathfrak A$. The subset $\overline{conv}(\mathfrak{S}_1\otimes \mathfrak{S}_2)$ in 
$\mathfrak S$ will be called the set of separable states and will be denoted by $\mathfrak{S}_{sep}$. The closure is taken with respect to the norm of $\mathfrak{A}^*$.
The subset $\mathfrak{S}\setminus \mathfrak{S}_{sep} \subset \mathfrak{S}$ is called the subset of entangled states.
\item \textit{$W^*$-algebra case, (cf. Theorem \ref{5.6}.)} \newline
Let $\mathfrak{M}_i$, $i=1,2$ be a $W^*$-algebra, $\mathfrak{M} = \mathfrak{M}_1 \otimes \mathfrak{M}_2$ be the spacial tensor product of $\mathfrak{M}_1$ and $\mathfrak{M}_2$, $\mathfrak S$ the set of all states on $\mathfrak M$, and $\mathfrak{S}^n$ the set of all normal states on $\mathfrak M$, i.e. the set of all normalized, weakly$^*$-continuous  positive forms on $\mathfrak M$ (equivalently, the set of all density matrices). The subset
$\overline{conv}^{\pi}(\mathfrak{S}_1^n\otimes \mathfrak{S}_2^n)$ in 
$\mathfrak S^n$ will be called the set of separable states and will be denoted by $\mathfrak{S}_{sep}^n$. \textit{The closure is taken with respect to the operator space projective norm on $\mathfrak{M}_{1,*} \odot \mathfrak{M}_{2,*}$.}
The subset $\mathfrak{S^n}\setminus \mathfrak{S}_{sep}^n \subset \mathfrak{S^n}$ is called the subset of normal entangled states.
\end{itemize}
\end{defn}
 
\begin{rem}
\begin{itemize}
\item As a separable state has the form of an arbitrary classical state (cf Corollary \ref{3.1}), it is naturally to adopt the convention that $\mathfrak{S}_{sep}$ ($\mathfrak{S}_{sep}^n$) contains only classical correlations. On the other hand, the set of entangled states is the set where the quantum (so extra) correlations can occur.
\item The difference between $C^*$-algebra case and $W^*$-algebra case stems from the Grothendieck's theory of tensor products. In particular, see Theorem \ref{5.6} and note how naturally the projective tensor product is appearing in Definition \ref{6.1}. Furthermore, to appreciate the use of
the topology defined by the operator space projective norm, we recall that the set of all normal states on $\mathfrak M$ is weakly-$*$ dense in the set of all states on $\mathfrak M$; see Example 4.1.35(2) in \cite{BR}.
\item The principal significance of normality of a state, from physical point of view, follows from existence of the number operator, see eg Theorem 5.2.14 in \cite{BR} and remarks given prior to this theorem.
In other words, a normal state corresponds to such situation when the number of particles has well defined sense.
\end{itemize}
\end{rem}

Having defined separable, and entangled states we turn to the question why other arguments given prior to Corollary \ref{3.1} are not working in the non-commutative setting. Our first observation is the following:

\begin{fc}
\label{6.4}
\begin{enumerate}
\item \textit{classical case.} \newline
Let $\delta_a$ be a Dirac's measure on a product measure space, i.e.
$\delta_a$ is given on $\Gamma_1 \times \Gamma_2$. Note that the marginal of the point measure $\delta_a$ gives another point measure, i.e. $\delta_a|_{\Gamma_1}=\delta_{a_1}$. Here we put
$a\in \Gamma_1 \times \Gamma_2$, $a=\left(a_1,a_2\right)$.  The same in ``physical terms'' reads: \textit{ a reduction of a pure state is again a pure state}.
\item \textit{non-commutative case.}\newline
Let $\mathcal{H}$ and $\mathcal{K}$ are finite dimensional Hilbert spaces. Without loss of generality we can assume that dim$\mathcal{H}$=dim$\mathcal{K} = n$. 
Let $\omega_x( \cdot) = (x,\cdot\ x)$ be a state on  $B(\mathcal{H})\otimes B(\mathcal{K})$ where $x$ is assumed to be of the form
\begin{equation}
x=\frac{1}{\sqrt{n}}\left(\sum_i e_i\otimes f_i\right)
\end{equation}
Here $\{e_i\}$ and $\{f_i\}$ are basis in $\mathcal{H}$ and $\mathcal{K}$ respectively.
Then, we have
\begin{equation}
\begin{split}
\omega_x\left(A\otimes \mathds{1}\right)&=\frac{1}{n}\left(\sum_i e_i\otimes f_i, A\otimes \mathds{1} \sum_j e_j\otimes f_j\right)\\
&=\frac{1}{n}\sum_{i,j}\left(e_i,Ae_j\right)\left(f_i,f_j\right)= \Tr_{\mathcal{H}}\frac{1}{n}\mathds{1}A \equiv \Tr_{\mathcal{H}}\varrho_0 A,
\end{split}
\end{equation}
where $\varrho_0 = \frac{1}{n}\mathds{1}$ is
``very non pure'' state. In other words, \textit{the non-commutative counterpart of the marginal of a point measure (pure state) does not need to be again a point measure (pure state)}. Consequently, the crucial ingredient of the discussion leading to Corollary \ref{3.1} is not valid in non-commutative case.
\end{enumerate}
\end{fc}

The second difficulty follows from the geometrical characterization of the set of states. Namely,
in geometrical description of a convex set in finite dimensional spaces one can distinguish two kinds of convex closed sets: 
simplexes and non-simplexes. Namely, let $K$ be a convex compact set. 
From Krein-Milman theorem, the set $K$ has extreme points $\{k_i\}$ 
and $K=\overline{conv}\left\{k_i\right\}$. Thus $K$ 
is a convex hull of its extreme points $\{k_i\}$. If any point of $K$ can be given uniquely as convex combination of extreme points
then $K$ is called a simplex. For example, a triangle is a simplex, but a circle is not a simplex. This gives an intuition. However,
here, considering infinite dimensional algebras and being interested in certain subsets of positive forms on these algebras, the definition of Choquet is more suitable.
Namely, let $K$ be a base of a convex cone $C$ with apex at the origin. The cone $C$ gives rise to the order $\leq$ ($a\leq b$ if and only if $b - a \in C$).
$K$ is said to be a simplex if $C$ equipped with the order $\leq$ is a lattice, see \cite{choquet} or \cite{BR} for details. (Lattice is a partially ordered set in which every two elements have a supremum and an infimum).
The importance of this notion follows from the well known result saying that in the classical case, the set of all states forms a simplex while this is not true for the quantum case. More precisely, see Example 4.2.6 in \cite{BR}
\begin{prop}
Let $\mathfrak{A}$ be a C*-algebra. Then the following conditions are equivalent
\begin{enumerate}
\item The state space $\mathfrak{S}_{\mathfrak{A}}$ is a simplex.
\item $\mathfrak{A}$ is abelian algebra.
\item Positive elements $\mathfrak{A}^+$ of $\mathfrak{A}$ form a lattice.
\end{enumerate}
\end{prop}
Therefore in quantum case the set of states is not a simplex (contrary to the classical case). Consequently,
in quantum case, all possible decompositions of a given state should be taken into account. More details on decomposition theory will be given in the next section.

We wish to close this section with a basic package of terminology used in open systems  and quantum information theory (so also in an analysis of quantum composite systems).
A linear map $\alpha : B(\mathcal H) \to B(\mathcal K)$ is called $k$-positive if a map $id_{M_k} \otimes \alpha : M_k \otimes B(\mathcal H) \to M_k \otimes B(\mathcal K)$ is positive, where $M_k \equiv M_k(\mathbb{C})$ denotes the algebra of $k\times k$ matrices with complex entries. A map $\alpha$ is called completely positive if it is $k$-positive for any $k$. A completely positive map $\alpha$ will be shortly called a CP map.
A positive map $\alpha:B(\mathcal H) \longrightarrow B(\mathcal K)$ is called
\textit{decomposable}
if
there are completely positive maps $\alpha_1,\alpha_2: B(\mathcal H)\longrightarrow B(\mathcal K)$
such that $\alpha=\alpha_1+\alpha_2\circ\tau_{\mathcal H}$, where $\tau_{\mathcal H}$ stands for a transposition map on $B(\mathcal H)$. Let $\mathcal P$, $\mathcal{P}_c$ and $\mathcal{P}_d$
denote the set of all positive, completely positive
and decomposable maps from $B(\mathcal H)$ to $B(\mathcal K)$, respectively. Note that
$$ 
\mathcal{P}_c \subset \mathcal{P}_d \subset \mathcal P
$$ 
(see also \cite{Ch1}- \cite{Ch2}, \cite{MM}).

Finally, let us define the family of \textit{PPT} (transposable) states on
$B(\mathcal H) \otimes B(\mathcal K)$
\begin{equation}
\mathfrak{S}_{PPT}=\{\varphi\in\mathfrak{S}: \varphi\circ(id_{B(\mathcal H)} \otimes \tau_{\mathcal K})\in \mathfrak{S}\}.
\end{equation}
where, as before, $\tau_{\mathcal K}$ stands for the transposition map, now defined on $B(\mathcal K)$.
Note that due to the positivity of the transposition $\tau_{\mathcal K}$ every
separable state $\varphi$ is transposable, so
$$ 
\mathfrak{S}_{sep}\subset\mathfrak{S}_{PPT}\subset\mathfrak{S}.
$$ 

\section{Decomposition theory}
This section provides a brief exposition of the decomposition theory; for proofs and further details we refer the reader to  Chapter IV in \cite{BR}, or to Chapter 6 in \cite{choquet} or to \cite{Alfsen}.
In the preceding section we note that a state of a quantum system can be decomposed in many ways. We wish to study this question in details.
The general idea of decomposition theory, applied to a convex compact subset $K$ of states, $K\subseteq \mathfrak S$, is to express 
the complex structure of $K$ as a sum of more simpler compounds. To this end, we wish to find a measure 
$\mu$ which is supported by extremal points $Ext({K})$ of $K$ and which decomposes a state $\omega \in {K}$ 
in the form 

\begin{equation}
\label{47}
\omega(A)=\int_{{K}}\omega'(A)d\mu (\omega')
\end{equation}\noindent
where $A$ is an observable. To comment the above strategy, firstly, we recall that for the classical case (see Section 3) Dirac's measures (so pure states) played an important role
in arguments leading to Corollary \ref{3.1}. This explains why we will be interested in a decomposition of a quantum state into convex combination of pure states. Thus, 
the claim that the measure $\mu$ should be supported by the subset $Ext(K)$ should be clear.
Secondly, to understand the framework of the strategy we note (cf Rules 1-5) that with a system we associate the pair $(\mathfrak{A}, \mathfrak S)$ consisting of the specific ($C^*$)-algebra
$\mathfrak A$ and the set of states $\mathfrak S$. But $(\mathfrak{A}, \mathfrak{S})\subset(\mathfrak{A}, \mathfrak{A}^*)$ and $(\mathfrak{A}, \mathfrak{A}^*)$ forms a dual pair. Hence $\omega'(A)$
can be considered as a function $\mathfrak{S} \supseteq K\ni \omega' \mapsto \hat{A}(\omega') \equiv \omega'(A)$. Moreover, $\hat{A}(\cdot)$ can be considered as an affine function on $K$. Hence, (\ref{47}) can be rewritten as

\begin{equation}
\label{48}
\hat{A}(\omega)=\int_{{K}}\hat{A}(\omega')d\mu (\omega')
\end{equation}
This indicates that, in fact, we are studying the barycentric decompositions, where the barycenter of a measure $\mu$ is defined as
\begin{defn}
\label{7.1}
Let $K$ be a compact convex subspace in locally compact space $X$ and let $\mu$ be a positive non-zero 
measure on $K$. We say that
\begin{equation}
b({\mu})=\mu\left(K\right)^{-1}\int_K xd\mu(x)
\end{equation}
is a barycenter of a measure $\mu$, where the integral is understood in the weak sense.
\end{defn}\noindent
We will need
\begin{defn}
Let a $C^*$-algebra $\mathfrak A$ be given. A state
$\omega \in \mathfrak{S}_{\mathfrak A}$ is a
 pure state if and only if positive linear functionals $\varphi$ majorized by $\omega$ (that is $\omega - \varphi$ is still positive)
are of the form $\lambda\omega$ where $\lambda\in \left(0,1\right]$.
\end{defn}\noindent
and
\begin{prop}
Let $\mathfrak{A}$ be a C*-algebra (not necessary with unit) and let $B_{\mathfrak{A}}$ denote positive linear 
functionals on  $\mathfrak{A}$ with norm less than or equal to one. Then $B_{\mathfrak{A}}$ is a convex, weakly 
$*$-compact subset of the dual  $\mathfrak{A}^*$ whose extremal points are pure states. 
\end{prop}\noindent
and also
\begin{prop}
The set of states $\mathfrak{S}_{\mathfrak{A}}$ is convex but it is 
weakly $*$-compact if and only if $ \mathfrak{A}$ has a unit. In the latter case the extreme 
points of $\mathfrak{S}_{\mathfrak{A}}$ are pure states. Thus, it follows from Krein-Milman theorem, that $\mathfrak{S}_{\mathfrak{A}}=\overline{conv}(\mathfrak{S}_{\mathfrak{A}}^P)$ where $\mathfrak{S}_{\mathfrak{A}}^p$
stands for the set of all pure states in $\mathfrak{S}_{\mathfrak{A}}$.
\end{prop}
It is of interest to note, see Corollary I.2.4 in \cite{Alfsen}
\begin{prop}
A point $\omega$ in a compact convex set $K$ is extreme if and only if $\delta_{\omega}$ (so the point measure) is the only measure in $M_1(K)$ with barycenter $\omega$, where $M_1(K)$ denotes the set of probability measures on $K$.
\end{prop}

To have compactness of the set of states, in the sequel, we will assume that the specific algebra associated with a system possesses the unit. It should be noted that for algebras appearing in physical models, this condition is satisfied.

 Having clarified the role of pure states, so also the choice of $Ext(K)$ in the above description of decomposition strategy,  let us turn to the measure theoretical 
aspects of the description of $\mathfrak{S}_{\mathfrak{A}}$.
At the first guess it appear reasonable that measurability of $Ext(\mathfrak{S}_{\mathfrak{A}})$ would 
be sufficient to ensure that one can find reasonable measures such that $\mu\left(\mathfrak{S}_{\mathfrak{A}}^p\right)=1$ i.e. the measure 
$\mu$ is supported by $Ext(\mathfrak{S}_{\mathfrak{A}})$. However, surprisingly, this is not the case. 
To overcome this problem, an additional condition as well as the concept of pseudosupported measure must be introduced. But, this involves the Baire structure.

Let us clarify this point.
Let $X$ be a locally compact convex set (for example, the space of states). There is a well defined Borel structure (cf \cite{H}) and Baire 
structure on $X$. Namely, Borel sets of a locally compact Hausdorff space $X$ are elements of $\sigma$-algebra generated by all compact subsets of $X$. Thus
Borel sets of $X$ contain all countable unions of closed sets, ${F}_{\sigma}$,
 as well as all countable intersections of open sets, ${G}_{\sigma}$. Note that this result follows on axioms of $\sigma$-algebra.

Therefore, the following definition of Baire structure is natural:
\begin{defn}
Baire sets of $X$ are defined as elements of $\sigma$-algebra $\mathcal F$ generated by compact (${G}_{\sigma}$)
subsets.
\end{defn}\noindent
To appreciate this concept we note (cf Section 4.1.2 in \cite{BR}).
\begin{rem}
Assume $K$ is a compact convex subset of a real locally convex topological vector space. Then
\begin{enumerate}
\item
The Baire sets of the compact space $K$ are defined as elements of the $\sigma$-algebra $\mathcal F$ generated by the closed $G_{\sigma}$-subsets, or the open $F_{\sigma}$-subsets of $K$.
\item
$\mathcal F$ is the smallest $\sigma$-algebra such that all continuous functions are measurable.
\end{enumerate}
\end{rem}\noindent
Now we are in position to give, cf Section 4.1.2 in \cite{BR}:
\begin{defn}
\begin{enumerate}
\item Let $M_+(K)$ denote the set of all positive (Radon) measures on $K$, (K is a compact set). The support of measure $\mu \in M_+(K)$ is defined as the smallest 
closed subset $C$  of $K$ such that $\mu(C)=\mu(K)$.
\item The measure $\mu$ is said to be pseudosupported by an arbitrary set 
$A\subseteq K$ if $\mu(B)=0$ for all Baire sets $B$ such that $B\cap A=\emptyset$
\end{enumerate}
\end{defn}
Although the concepts of support and pseudosupport look similarly, one can provide an example indicating that the similarity can be confusing (cf Section 4.1.2 in \cite{BR}). Namely,
\begin{exmp}
One can find a probability measure $\mu$ and a Borel subset $A$ such that  
$\mu$ is pseudosupported by $A$, but $\mu(A)=0$.
\end{exmp}
This example illustrates ``unpleasant'' measure-theoretical case which will demand a special care in the strategy to decompose the subset $K \subseteq \mathfrak{S}_{\mathfrak A}$.

After these preliminaries we turn to more detailed presentation of barycentric decomposition. To this end,
let $K$ be a compact convex set, and $M_1(K)$ be the set of positive normalized, i.e probability, measures on $K$, which is a (weakly-$*$) compact. 
To proceed with a description of barycentric decomposition, let us take an arbitrary element $\omega$ of $K$ and define the collection $M_{\omega}(K)$ of probability measures $\mu$ having 
as the barycenter the fixed element $\omega$ of $K$:

\begin{equation}
M_{\omega}(K)=\left\{\mu \in M_1(K):\ \int_{K}\omega'd\mu(\omega')=\omega   \right\}
\end{equation}

\begin{rem}
Note that $\mu \in M_{\omega}(K)$ is equivalent to the statement : $\mu 
\sim \delta_{\omega}$, i.e. the measure $\mu$ is equivalent to the Dirac measure $\delta_{\omega}$, where the equivalence of probabilistic measures measures $\mu$ and $\mu^{\prime}$ is defined as
$$ \int_K \nu d\mu(\nu) = \int_K \nu d\mu^{\prime}(\nu).$$
We remind, the integrals are understood in the weak sense, (see \cite{Alfsen} for details).
\end{rem}
The following result (see Section 4.1.2 in \cite{BR}; cf also the formula (\ref{48})) is saying that the barycenter $b(\mu)$ of a general measure  $\mu$ exists.
\begin{prop}
For a $\mu\in M_1(K)$ there exists a unique point b($\mu$) in the set $K$ such that 
\begin{equation}
f(b(\mu))=\int_{K}f(\omega')d\mu(\omega')
\end{equation}
for all affine, continuous, real-valued functions $f$ on $K$.
\end{prop}\noindent
To convince the reader that there are nontrivial decompositions, we need (see Definition 4.1.2 in \cite{BR})
\begin{defn}
Let $K$ be a compact set. 
On the set of (all, Radon) positive measures on K, $M_+(K)$, the relation $\geq$ can be defined in the following way: $\mu \geq \nu$ if
$$ \mu(f) \geq \nu(f) $$
for all real continuous convex functions $f$ on $K$.
\end{defn}\noindent
To appreciate this relation, we note (see Proposition 4.1.3 in \cite{BR})
\begin{prop}
\label{7.13b}
\begin{enumerate}
\item The relation $\geq$ on $M_+(K)$ is a partial ordering.
\item $\mu \in M_{\omega}(K)$ if and only if $\mu \geq \delta_{\omega}$.
\item Each $\omega \in K$ is the barycenter of a measure $\mu \in M_1(K)$ which is maximal for the order $\geq$.
\end{enumerate}
\end{prop}\noindent
and (see Theorem 4.1.15 in \cite{BR})
\begin{prop}
\label{7.13a}
Let $K$ be a convex compact subset of a locally convex Hausdorff space. The following two conditions are equivalent:
\begin{enumerate}
\item each $\omega \in K$ is the barycenter of a unique maximal measure,
\item $K$ is a simplex
\end{enumerate}
\end{prop}\noindent
Note that as we are interested in convex compact sets which are not simplexes, the above results guarantee a nontrivial decomposition of any non-pure state in $K$.

On the other hand, results presented in this section clearly indicate that the measure-theoretic properties of the extremal points $Ext(K)$, of $K$, are of fundamental importance.
Therefore, we turn to such situation in which subsets of points in $K \subseteq \mathfrak S$ have good measurability properties as well as the imposed conditions on $K$ are acceptable from ``physical'' point of view.
This is nontrivial task as for a general $C^*$-algebra $\mathfrak A$, the set of pure states may be as pathological as the set of extremal points of a general convex compact set, see Section  4.1.4 in \cite{BR}.

 Thus, to avoid pathological properties of $Ext(K)$ we will
make the following assumption, (see \cite{ruelle1}, \cite{ruelle2}, and Definition 4.1.32 in \cite{BR}).

\begin{defn}\textit{ Ruelle's SC condition}\newline
\label{7.8}
Let $\mathfrak{A}$ be a C*-algebra with unit, and $\mathfrak{F}$ a subset of the state space $\mathfrak{S}_{\mathfrak{A}}$. 
$\mathfrak F$ is said to satisfy separability condition (SC) if there exists a sequence of sub-C*-algebras 
$\left\{\mathfrak{A}_n\right\}$ such that $\bigcup_{n=1}^{\infty}\mathfrak{A}_n$ is dense in $\mathfrak{A}$ 
and each $\mathfrak{A}_n$ contains a two-sided, closed, separable ideal $\mathcal{I}_n$ such that 
$$\mathfrak{F}=\left\{\omega,\ \omega\in \mathfrak{S_{\mathfrak{A}}},\ \left\|\omega|_{\mathcal{I}_n}\right\|=1,\ n\geq 1 \right\}.$$
\end{defn}
 
To illustrate widespread applications of SC-condition we give several examples having clear ``physical'' interpretation. Note, that these examples also illustrate the phrase ``specific algebras'' used in Rule 4.1.
\begin{exmp}
\label{sepalgebra}
\begin{enumerate}
\item Let a $C^*$-algebra
$\mathfrak{A}$ be  separable one. Put $\mathcal{I}_n=\mathfrak{A}$ for any $n$. 
Then SC-condition is fulfilled. We emphasize that 
in this case the weak-$*$-topology on the state space $\mathfrak{S}_{\mathfrak{A}}$ is metrizable. Then, Borel and Baire 
structure coincides. Consequently, measure-theoretic description of $\mathfrak{S}_{\mathfrak{A}}$ is greatly simplified.
\item
Obviously, the basic algebra in quantum information literature $\mathfrak{A} \equiv M_{n}(C)$ is a nice example of a separable algebra.
\end{enumerate}
\end{exmp}

The next example is relevant for Dirac's formalism of quantum mechanics. 
\begin{exmp}
\label{dirac's}
Let
$\mathfrak{A}$ be equal to $B(\mathcal{H})$ and suppose that  $\mathcal{H}$ is infinite dimensional Hilbert space. 
By $\mathcal{F}_C(\mathcal{H})$ we denote the space of compact operators on $\mathcal H$ and 
$\mathcal{F}_T(\mathcal{H})$ denotes subspace of trace class operator.
We have following isomorphisms, see \cite{schatten} 
\begin{equation}
\left(\mathcal{F}_C(\mathcal{H})\right)^*=\mathcal{F}_T(\mathcal{H})
\end{equation}
\begin{equation} 
\left(\mathcal{F}_T(\mathcal{H})\right)^*=B(\mathcal{H}) 
\end{equation}
Moreover, see \cite{schatten}, both $\mathcal{F}_C(\mathcal{H})$ and $\mathcal{F}_T(\mathcal{H})$ are two-sided ideals in $B(\mathcal{H})$. Normal states (so density matrices) $\mathfrak{S}_{B(\mathcal{H})}^n \equiv \mathfrak{S}^n$(positive, weakly$^*$-continuous, normalized forms on $B(\mathcal{H})$) have the following characterization (cf Proposition 2.6.14 in \cite{BR})
\begin{equation}
\mathfrak{S}^n=\left\{\omega \in \mathfrak{S}_{B(\mathcal{H})},\ \left\|\omega|_{\mathcal{F}_C}\right\|=1\right\}  
\end{equation}
Obviously, the specification:
$\mathfrak{A}=B(\mathcal{H})$, $\mathfrak{A}_n=B(\mathcal{H})$ and $\mathcal{I}_n=\mathcal{F}_C(\mathcal{H})$ 
in Definition \ref{7.8} shows that the condition SC is fulfilled for $\mathfrak{S}^n$. 
\end{exmp}
It is worth pointing out that one can say more. 
\begin{rem}
On the set $\mathfrak{S}^n$ (so on the set of all density matrices)\newline the weak-$*$ and 
uniform topology (of the predual) coincide, see Proposition 2.6.15 in \cite{BR}. In other words, weak-$*$-topology is metrizable on $\mathfrak{S}^n$.
Consequently, Baire and Borel structures coincide. However, this is not true for the set of all states $\mathfrak{S}_{B(\mathcal{H})}$. In particular, $\mathfrak{S}_{B(\mathcal{H})}$, for the weak-$^*$ topology is not a metric space.
\end{rem}\noindent
The next example is relevant to lattice models, models of solid state physics, spin chains.

\begin{exmp} \textit{UHF - uniformly hyperfinite algebra.}\newline 
\label{7.15}
$\mathbb{Z}^n$ is the Cartesian product of integer numbers $\mathbb{Z}$ with $n=\left\{1,\ 2,\ 3,\ \ldots N\right\}$.  Let $\alpha\in \mathbb{Z}^n$ be an arbitrary, fixed site of the lattice. With each site $\alpha$ we associate a Hilbert space $\mathcal{H}_{\alpha}$. It is required that each Hilbert space $\mathcal{H}_{\alpha}$ is finite dimensional.
The Hilbert space associated with a finite subset $\Lambda\subset Z^n$ is given by 
$\mathcal{H}_{\Lambda}=\otimes_{\alpha\in \Lambda}\mathcal{H}_{\alpha}$.
We put $\mathfrak{A}_{\Lambda}=B(\mathcal{H}_{\Lambda})$. The algebra of operators associated with the whole 
space $\mathbb{Z}^n$ is equal to $\mathfrak{A}=\overline{\bigcup_{\Lambda}\mathfrak{A}_{\Lambda}}$.
It is easy to check that SC is satisfied. \newline

The just described model, with the algebra $\mathfrak{A}=\overline{\bigcup_{\Lambda}\mathfrak{A}_{\Lambda}}$, has unexpected property, see Example 4.1.31 in \cite{BR}. Namely, its set of pure states $Ext(\mathfrak{S}_{\mathfrak{A}})$ 
is a weakly-$*$ dense subset in $\mathfrak{S}_{\mathfrak{A}}$. Consequently, in general, the set  $Ext(\mathfrak{S}_{\mathfrak{A}})$  does not 
have to be closed.
\end{exmp}\noindent
The preceding example can be generalized.

\begin{exmp}\textit{Quasi local algebras.}\newline
\label{7.13}
Replace $\mathbb{Z}^n$ by $\mathbb{R}^n$, i.e we replace integer numbers $\mathbb{Z}$ by real numbers $\mathbb{R}$.
In physical terms it means that now we are interested in continuous systems. We associate a Hilbert  space $\mathcal{H}_{\Lambda}$  (usually infinite dimensional) with the region (bounded subset) $\Lambda\subset \mathbb{R}^n$. 
Put  $\mathfrak{A}_{\Lambda}=B(\mathcal{H}_{\Lambda})$. It means that with a region $\Lambda$ in $\mathbb{R}^n$ we associate \textit{local} observables and these observables generate the (specific) algebra $\mathfrak{A}_{\Lambda}$. For $\Lambda\subset \Lambda'$ we have $\Lambda'=\Lambda \cup (\Lambda' \setminus \Lambda)$ and according to Rule 4.4
one has $\mathcal{H}_{\Lambda'}=\mathcal{H}_{\Lambda}\otimes \mathcal{H}_{\Lambda'\setminus \Lambda}$. Then, the algebra of all
observables is given by $\mathfrak{A}=\overline{\bigcup_{\Lambda}\mathfrak{A}_{\Lambda}}$. 
\end{exmp}\noindent
To say more on this model we need:

\begin{defn}
A state $\omega\in \mathfrak{S}_{\mathfrak{A}}$ is called locally normal if $\omega|_{\mathfrak{A}_{\Lambda}}$ is normal for every $\Lambda$.
\end{defn}

Thus, a state is locally normal, if its restriction to an algebra associated with any bounded region is given by a (local) density matrix. Let us observe that
for each $\Lambda$ in $\mathbb{R}^n$ one has $\mathcal{F}_C(\mathcal{H}_{\Lambda})\subset B(\mathcal{H}_{\Lambda})$.
Therefore, the condition SC holds for this model.
Moreover, one has, see Theorem 2.6.16 in \cite{BR}
\begin{prop}
Denote by $\mathfrak{S}_l$ the set of locally normal states on the algebra $\mathfrak A$ defined in Example \ref{7.13}.
The set $\mathfrak{S}_l$ is metrizable in weak-$*$ topology.
\end{prop}
Consequently, the measure-theoretic properties of the set of locally normal states for quasi local algebra described in Example \ref{7.13} are simplified!

Example \ref{7.13} still can be generalized. The motivation of that follows from Rule 1: observables of a system associated with a region $\Lambda \subseteq \mathbb{R}^n$ do not need to generate the
algebra $B(\mathcal{H}_{\Lambda})$. We emphasize that this problem is in ``the heart'' of Haag-Kastler approach to Quantum Field Theory, see \cite{HK}.
\begin{exmp}
\label{7.19}
Let $\mathbb{R}^n$ be a ``coordinate'' space of a field and $\Lambda\subset \mathbb{R}^n$ be a region. With the region $\Lambda$ we associate a C*-algebra $\mathfrak{A}_{\Lambda}$ (which is not necessary of the form $B(\mathcal{H}_{\Lambda})$). However, we are assuming that conditions given in Definition \ref{7.8} hold for the family $\{ \mathfrak{A}_{\Lambda} \}_{\Lambda \subset \mathbb{R}^n}$. Put
$\mathfrak{A}=\overline{\bigcup_{\Lambda}\mathfrak{A}_{\Lambda}}$. Obviously, the definition of $\mathfrak{ A}$ guarantees that the condition SC is fulfilled.
\end{exmp}
\begin{rem}
The extremely important point to note here is that Examples \ref{7.13}, \ref{7.19} as well as \ref{7.15}
give a possibility to discuss how far one subsystem is  from another. Consequently, this gives a chance to speak about propagation of an effect. Thus, the proper framework for a discussion of a form of causality is provided. Moreover, quasi local algebras are basic ones in quantum field theory; for example, a quasi local algebra is an essential ingredient of the Reeh-Schlieder theorem (see also the end of Section 4.
\end{rem}
Having clarified measure-theoretic aspects of decomposition theory we wish to state the main result. To this end we need:
\begin{defn}
A face $F$ of a compact convex set $K$ is defined  to be a convex subset of $K$ with the following property: if $\omega\in F$ can be written as $\omega=\lambda_1\omega_1+\lambda_2\omega_2$ where $\lambda_i\geq 0$, $i=1,2$,  $\lambda_1+\lambda_2=1$, $\omega_1$ and $\omega_2$ are in $K$ then we have $\omega_i\in F$ ($i=1,2$).
\end{defn}\noindent

\begin{exmp}
\begin{enumerate}
\item One point face, $F=\{ \omega \} \subset K$ is an extremal point of $K$.
\item Let $\mathfrak M$ be a von Neumann algebra. Then $\mathfrak{S}^n_{\mathfrak M}$ is a face in
$\mathfrak{S}_{\mathfrak M}.$
\end{enumerate}
\end{exmp}
The promised result is, see Theorem 4.2.5 in \cite{BR}:
\begin{thm}
\label{7.22}
Let $\mathfrak{A}$ be a C*-algebra with identity and $\omega$ be a state over $\mathfrak{A}$.
There are measures $\mu$ (determined by structures induced by the state $\omega$) over $\mathfrak{S}_{\mathfrak{A}}$ such that any $\mu$ is pseudosupported by pure states 
$Ext(\mathfrak{S}_{\mathfrak{A}})$. Moreover, if additionally $\omega$ is in a face $F$ of $\mathfrak{S}_{\mathfrak{A}}$ 
satisfying the separability condition SC then the set of extremal points $Ext(F)$ of $F$, is a Baire subset of the pure states on $\mathfrak{A}$ and 
$$\mu(Ext(F)) =1.$$
\end{thm}

\begin{rem}
\begin{enumerate}
\item Obviously, the results stated in above  Theorem are significant for non trivial faces, i.e. when $F$ is consisting of more than one point.
\item Theorem \ref{7.22} says that the strategy described at the beginning of this section is working, i.e. a state $\omega \in F$ can be decomposed into pure states.
\item Note, we do not claim that for any $\omega \in F \subseteq \mathfrak{S}_{\mathfrak{A}}$ there is a unique measure giving the desirable decomposition. In fact, it would be equivalent that $\mathfrak A$ is abelian, cf Proposition \ref{7.13a}; see also page 360 in \cite{BR}.
\item Measures appearing in Theorem \ref{7.22} are of very special type. They are in the class of so called \textit{orthogonal measures}, see Section 4.2 in \cite{BR} for details.
\item We remind that in Examples: in \ref{sepalgebra} the set of all states; in \ref{dirac's} the subset of all normal states; in \ref{7.13} the subset of all locally normal states; all the mentioned subsets have simplified measure-theoretic description. Namely, for these subsets, weak-$*$ topology is metrizable. In particular, pseudosupporting
is equivalent to the notion of support.
\end{enumerate}
\end{rem}

Instead of describing orthogonal measures in details,  we restrict ourselves to the following enlightening example:
\begin{exmp}
\label{7.24}
To be specific, consider
$\mathfrak{A}=B(\mathcal{H})$ with finite dimensional Hilbert space $\mathcal H$, thus dim$\mathcal{H}< \infty$. Any state $\omega$ on $\mathfrak{A}$, has the following form: $\omega(\cdot)=Tr(\rho \ \cdot)$ where $\rho$ is a positive operator on $\mathcal H$ with trace equal to $1$, $Tr \rho = 1$ (a density matrix). Now, we wish to construct GNS
representation of $B(\mathcal{H})$ induced by the state $\omega$ (cf Section 2.3.3 in \cite{BR}). To this end,
$B(\mathcal{H})$ is equipped with the scalar product
$$B(\mathcal{H}) \times B(\mathcal{H}) \ni \left\langle A,B\right\rangle \mapsto \Tr A^*B \in \mathbb{C} $$
The completion of $B(\mathcal{H})$ with respect to the above scalar product gives the Hilbert space $\mathcal{H}_{\omega}.$
It is worth noting that in the Hilbert space $\mathcal{H}_{\omega}$, there exists specific vector $\rho^{\frac{1}{2}} \in \mathcal{H}_{\omega}$,
such that the set $\{ A \rho^{\frac{1}{2}}; \quad A \in B(\mathcal H) \}$ is dense in $\mathcal{H}_{\omega}.$
Such vector is called the cyclic one and will be denoted by $\Omega_{\omega}$. Moreover,
\begin{equation}
\omega(A)=(\Omega_{\omega},A\Omega_{\omega}) \equiv \Tr \rho^{\frac{1}{2}} A \rho^{\frac{1}{2}},
\end{equation}
where $A \in B(\mathcal H)$.
Let us define two (different in general) representations of $B(\mathcal H)$
\begin{equation}
\pi_l(A)\psi=A\psi
\end{equation}
and
\begin{equation}
\pi_r(A)\psi=\psi A
\end{equation}
where $A \in B(\mathcal H)$, $\psi \in \mathcal{H}_{\omega}$.
Thus, the representation $\pi_l$ is defined as the left multiplication while $\pi_r$ is defined as the right multiplication.
Let us observe that $\pi_l(B(\mathcal{H}))\equiv \{ \pi_l(A); A \in B(\mathcal{H})\}$ and similarly  defined 
$\pi_r(B(\mathcal{H}))$ are subsets of $B(\mathcal{H}_{\omega})$.
One can say more. Namely $\pi_l(B(\mathcal{H}))$ and $\pi_r(B(\mathcal{H}))$ commute with  each other as operators in $B(\mathcal{H}_{\omega})$, i.e. $AB = BA$ for $A \in \pi_l(B(\mathcal{H}))$ and $B \in \pi_r(B(\mathcal{H}))$.
To fix terminology, the triple $(\mathcal{H}_{\omega}, \pi_{\omega}(\cdot), \Omega_{\omega})$ is called the GNS representation, where to follow the standard notation we put $\pi_{\omega}(\cdot) = \pi_l(\cdot)$. 

Let $\mathfrak B$ be the abelian subalgebra in $\pi_r(B(\mathcal{H}))$ generated by the maximal collection of pairwise orthogonal projectors $\{p^{\pi}_i\}_i$ in $\pi_r(B(\mathcal{H}))$, i.e. 
$\{p^{\pi}_i\}_i$ is the subset of $ \pi_r(B(\mathcal{H}))\subset \pi_l(B(\mathcal{H}))'$ such 
that $p^{\pi}_ip^{\pi}_j=\delta_{ij}p^{\pi}_i$, $p^{\pi}_i = (p^{\pi}_i)^*$  and $\sum p^{\pi}_i=\mathds{1}$. $\pi_l(B(\mathcal{H}))'$ stands for the set of all operators in $B(\mathcal{H}_{\omega})$ which commute with operators in $\pi_l(B(\mathcal{H}))$
while maximality means that each $p^{\pi}_i$ can be taken as $p^{\pi}_i = \pi_r(|e_i><e_i|) \equiv \pi_r(p_i)$, where $e_i$ is a vector of norm one in $\mathcal H$, and finally the family of vectors $\{e_i\}$ forms a basis in $\mathcal H$.

One can observe that

\begin{equation}
\begin{split}
\omega(A)&= \Tr \left(\rho^{\frac{1}{2}}\cdot A\rho^{\frac{1}{2}}\right)=\left(\mathds{1}\rho^{\frac{1}{2}},\ A\rho^{\frac{1}{2}}\right)=\left(\left(\sum_ip^{\pi}_i\right)\rho^{\frac{1}{2}},\ A\rho^{\frac{1}{2}}\right)\\
&=\sum_i\left(\rho^{\frac{1}{2}}p_i,\ A\rho^{\frac{1}{2}}\right)=\sum_iTr\ p_i\rho^{\frac{1}{2}}A\rho^{\frac{1}{2}}
\end{split}
\end{equation}

Define $\widetilde{\omega_i}(A)=Tr\ p_i\rho^{\frac{1}{2}}A\rho^{\frac{1}{2}}$. 
Trace of a positive operator is a positive number, so $\widetilde{\omega_i}(\cdot)$ is a positive form.
$\widetilde{\omega_i}(\mathds{1})=Tr\ p_i\mathds{1}\rho=\lambda_i\geq 0$. 
Consequently, defining 
$\omega_i(\cdot)=\frac{1}{\lambda_i}\widetilde{\omega_i}(\cdot)$ for those $i$ that $\lambda_i \neq 0$  we have the following decomposition of a given state 
 
\begin{equation}
\omega (\cdot)=\sum_i\lambda_i\omega_i(\cdot)
\end{equation}
\end{exmp}
\begin{rem}
\begin{enumerate}
\item
The presented concept of decomposition (based on orthogonal measures) can be viewed as a generalization of spectral decomposition. 
\item The strategy of the construction of orthogonal measures stems from the deep Tomita's result relating existence of such measure with certain abelian subalgebra, see Sections 4.1.3 and 4.2.1 in \cite{BR}.
Consequently, the above recipe provides, in general, many such measures.
\end{enumerate}
\end{rem}
To sum up this section: as the family of states in quantum mechanics does not form a simplex, a state can be decomposed in many ways. A decomposition of a state can be realized using measure-theoretical approach
and \textit{relevant results were provided}. We noted that extreme points of some subsets of states can exhibit ``bad'' measure-theoretic properties.
To avoid such cases, an auxiliary condition, separability condition SC, should be imposed. Fortunately, all essential physical models, which were described in this section, satisfy this condition. Consequently, the program of decomposing of states can be carried out and the main tool is provided by Theorem \ref{7.22}.

\section {Quantum correlations}

In this section we provide  a detailed exposition of quantization of the principal measure of correlations, correlation coefficient, see Section 2. In particular, to measure quantum correlations, we will define and study the coefficient of quantum correlations. Moreover, we will indicate why the techniques described in Sections 4-7 are indispensable for that purpose.

But before going into details we emphasize that existence of correlations in the quantum theory, likewise in the classical case, is not equivalent to the existence of causal relations, cf Introduction. In other words,
correlation between two variables is not a sufficient condition 
to establish a causal relationship. It is worth citing Dirac's point of view on this question, see page 4 in \cite{dirac}:
\textit{Causality applies 
only to a system which is left undisturbed. If a system is small we can not observe it 
without producing serious disturbance and hence we can not expect to find any causal 
connection between results of our observations. Causality will still be assumed to apply to 
undisturbed system and the equations which will be set up to describe an undisturbed 
system will be differential equations expressing a causal connection between conditions at
one time and conditions at a latter time}.

The important point to note here is that an analysis of causality demands \textit{equations expressing a causal connection}.
To this end, one should say what an effect of propagation is like. On the other hand, in classical probability theory, cf Section 2, to handle (classical) correlations, such additional requirement is not necessary. We will adopt this point of view as the important hint for the quantization, and therefore we will study correlations only.

As the first step of quantization procedure we recall that the correlation coefficient for the classical case was given as (cf Definition \ref{2.11}.)

$$
C(X,Y)=\frac{E(XY)-E(X)E(Y)}{(E(X^2)-E(X)^2)^\frac{1}{2}(E(Y^2)-E(Y)^2)^\frac{1}{2}}
$$
 Let us rewrite this definition in the quantum context. Let a $C^*$-algebra  $\mathfrak A$ be a specific algebra of observables (cf Rule 4.1), 
 $\varphi$ a state on $\mathfrak A$ (cf Rule 4.3), and $A,A' \in \mathfrak{A}$ be observables. Let us replace the classical expectation value $E(X)$ by the quantum one $\varphi(A) \equiv \left\langle A\right\rangle$. We note

\begin{equation}
\left\langle \left(A- \left\langle A \right\rangle \right)^2 
\right\rangle = \left\langle A^2  \right\rangle -   \left\langle A  \right\rangle^2
\end{equation}
and
\begin{equation}
\left\langle \left(A- \left\langle A \right\rangle \right) \left( A'- \left\langle A' \right\rangle     \right) 
\right\rangle = \left\langle AA'  \right\rangle -   \left\langle A  \right\rangle         \left\langle  A'  \right\rangle 
\end{equation}
Consequently, in the quantum context, one can write
\begin{equation}
C_q(A,A')=\frac{\left\langle \left(A- \left\langle A \right\rangle \right) \left( A'- \left\langle A' \right\rangle     \right) \right\rangle } 
{\left\langle \left(A- \left\langle A \right\rangle \right)^2 \right\rangle^\frac{1}{2}
\left\langle \left(A'- \left\langle A' \right\rangle \right)^2 \right\rangle ^\frac{1}{2}}
\end{equation}
and this form of correlation coefficient agrees with that given in Omn\'es  book, see page 363 in \cite{omnes}. Finally, we note that an application of Schwarz inequality shows that $C_q(A,A') \in [-1, +1]$, so $C_q(A,A')$ is normalized.

As the second step we wish to show that the correlation coefficient, $C_q(A,A')$, can recognize the ``very entangled'' states. Namely one has, see page 363 in \cite{omnes}:

\begin{exmp}
Let us consider the composite system such that its algebra of observables is given by $B(\mathcal{H}) \otimes B(\mathcal{H})$ and we take a state $\varphi$ of the form $\varphi(\cdot) = \Tr (\varrho \ \cdot)$,
where $\varrho$ is a density matrix (on the Hilbert space $\mathcal{H} \otimes \mathcal{H}$). Let us select two observables of the form
$A=0\cdot P_{e_0}+1 \cdot P_{e_1}$ and the same for $A'$, where $P_e$ stands for the orthogonal projector on the vector $e$; to shorten notation we write $A$ instead of $A\otimes \mathds{1}$ and $A'$ instead of $\mathds{1}\otimes A' $. 
We want to find a special state which gives maximal value of $C_q$. Observe that the condition 
$C_q=1$ gives
\begin{eqnarray}\nonumber
0&=&\left(\left\langle A^2\right\rangle -2\left\langle A\right\rangle^2 + \left\langle A\right\rangle^2 \right)
\left(\left\langle (A')^2\right\rangle -2\left\langle A'\right\rangle^2 + \left\langle A'\right\rangle^2 \right)\\
&&-\left[ \left\langle AA'\right\rangle  -2\left\langle A\right\rangle \left\langle A'\right\rangle+ \left\langle A\right\rangle  \left\langle A'\right\rangle   \right]^2
\end{eqnarray}
or equivalently for our choice of $A$ and $A'$ ($A^2 = A$, etc)
\begin{equation}
0=\left\langle A\right\rangle \left\langle A'\right\rangle\left[ 1-  \left\langle A\right\rangle   -  \left\langle A'\right\rangle 
+2 \left\langle AA'\right\rangle   \right]
-  \left\langle AA'\right\rangle^2                                     \label{cor}
\end{equation}
Let us adopt the following convention: $\rho_{ij',kl'}=
\left\langle ij'|\rho| kl'\right\rangle$. Assuming additionally that $dim \mathcal{H} = 2$, so considering two dimensional case, one has  $\Tr\rho_{\varphi}AA'=\rho_{11,11}$, 
$\Tr\rho_{\varphi}\mathds{1}\otimes A'=(\rho_{11,11}+\rho_{01,01})$, 
$\Tr\rho_{\varphi}A\otimes \mathds{1}=(\rho_{11,11}+\rho_{10,10})$. The formula 
(\ref{cor}) can be rewritten as

\begin{equation}
0=  \left(\rho_{11,11}+\rho_{10,10}\right) \left(\rho_{11,11}+\rho_{01,01}\right) 
\left(1-\rho_{10,10}-\rho_{01,01}\right) -\rho_{11,11}^2   \label{cor2}                      
\end{equation}
One can define maximally entangled state by
\begin{equation}
\Psi=\frac{1}{\sqrt{2}}\left(\left|10 \right\rangle- \left|01\right\rangle   \right)
\equiv \frac{1}{\sqrt{2}}(e_1 \otimes e_0 - e_0 \otimes e_1).
\end{equation}
We put $\rho_{\Psi}=\left|\Psi \right\rangle \left\langle \Psi\right|$. Then, 
$(\rho_{\Psi})_{11,11}=0=(\rho_{\Psi})_{00,00}$, $(\rho_{\Psi})_{01,01}=\frac{1}{2}$, $(\rho_{\Psi})_{10,10}=\frac{1}{2}$.
Obviously, (\ref{cor2}) is fulfilled for the state $\rho_{\Psi}$. Thus, the state $\rho_{\Psi}$, where $\Psi$ is a maximally entangled vector, gives an example of maximal correlation coefficient for the observables $A$ and $A'$, i.e. $C_q(A,A') =1 $.
\end{exmp}\noindent
The next example is:

\begin{exmp}
Let $\omega$ be a separable state on $\mathfrak{A} = \mathfrak{A}_1 \otimes \mathfrak{A}_2$,
$$ \omega(\cdot) =   \sum_{i=1}^{n} \lambda_i (\omega^1_i \otimes \omega^2_i)(\cdot),   $$
where $\omega^k_i$, $k=1,2$, $i=1,2,...,n$ are states in $\mathfrak{S}_{\mathfrak{A}_k}$.
It is a simple matter to check that, in general,
$$\omega(A_1 \otimes \mathds{1} \cdot \mathds{1} \otimes A_2)\neq \omega(A_1 \otimes \mathds{1})\omega(\mathds{1} \otimes A_2) $$
for $A_k \in \mathfrak{A}_k$. Therefore, the state $\omega$ contains some correlations. However, as the state $\omega$ is separable one, these correlations are considered to be of classical nature only.
\end{exmp}

To sum up, the straightforward quantization of the correlation coefficient gives a device for finding the size of correlations. BUT, the coefficient $C_q$ is not able to distinguish correlations of quantum nature from that of classical nature. Thus, a new measure of quantum correlation should be introduced. This will be done by defining the coefficient of quantum correlations, see \cite{Maj2003}, \cite{Maj2004}. The basic idea to define ``pure'' quantum correlations is to ``subtract'' classical correlations. In other words we will look for the best approximation of a given state $\omega$ by separable states, cf Example \ref{exm6.1}. However, a given state $\omega$, in general, can possess various decompositions. Thus, to carry out the analysis of such approximations we should use the decomposition theory, cf Section 7.

Having clarified the basic idea we proceed with the study of coefficient of (quantum) correlations for a quantum composite system 
specified by $(\mathfrak{A} =  \mathfrak{A}_1 \otimes \mathfrak{A}_2, \mathfrak{S}_{\mathfrak{A}} )$,
where $\mathfrak{A}_i$ are $C^*$-algebras. Thus we will consider $C^*$-algebra case (cf Definition \ref{6.1}).
We begin with the definition of the
restriction maps 

\begin{equation}
(r_1\omega)(A)=\omega(A\otimes \mathds{1}) 
\end{equation}
\begin{equation}
 (r_2\omega)(B)=\omega(\mathds{1}\otimes B),
\end{equation}
where $\omega \in \mathfrak{S}_{\mathfrak{A}}$, $A \in \mathfrak{A}_1$, and $B \in \mathfrak{A}_2$. Clearly,
$r_i:\mathfrak{S}_{\mathfrak{A}}\rightarrow\mathfrak{S}_{\mathfrak{A}_i}$ and the restriction map $r_i$ is continuous (in weak-$*$ topology), $i =1,2$, see Proposition 4.1.37 in \cite{BR}.
We recall that the decomposition procedure is starting with a ``good'' measure on the state space $\mathfrak S$ (so from $M_{\omega}(\mathfrak{ S})$). Hence, 
let us take a ``good'' measure $\mu$ on $\mathfrak{S}_{\mathfrak{A}}$. Define 
\begin{equation}
\label{69}
\mu_i(F_i)=\mu(r_i^{-1}(F_i))
\end{equation}
 for $i=1,2$, where $F_i$ is a Borel subset in $\mathfrak{S}_{\mathfrak{A}_i}$.
 It is easy to check that the formula (\ref{69}) provides the well defined measures $\mu_i$ on $\mathfrak{S}_{\mathfrak{A}_i}$, $i=1,2$.
Having two measures $\mu_1$, $\mu_2$ on $\mathfrak{S}_{\mathfrak{A}_1}$, and $\mathfrak{S}_{\mathfrak{A}_2}$ respectively, we want to "produce" a new measure $\boxtimes \mu$
on $\mathfrak{S}_{\mathfrak{A}_1}\times \mathfrak{S}_{\mathfrak{A}_2}$. 
To this end, firstly, 
let us consider the case of finitely supported probability measure $\mu$:
\begin{equation}
\mu=\sum_{i=1}^N\lambda_i \delta_{\rho_i}
\end{equation}
where $\lambda_i\geq0$, $\sum_{i=1}^N\lambda_i=1$, and $\delta_{\rho_i}$ denotes the Dirac's measure. 
We define
\begin{equation}
\label{71}
\mu_1=\sum_{i=1}^N\lambda_i\delta_{r_1 \rho_i}
\end{equation}
and 
\begin{equation}
\label{72}
\mu_2=\sum_{i=1}^N\lambda_i\delta_{r_2 \rho_i}.
\end{equation}
Then
\begin{equation}
\label{73}
\boxtimes\mu=\sum_{i=1}^N\lambda_i\delta_{r_1 \rho_i}\times \delta_{r_2 \rho_i} 
\end{equation}
provides a well defined  measure on $\mathfrak{S}_{\mathfrak{A}_1}\times \mathfrak{S}_{\mathfrak{A}_2}$. Here $\mathfrak{S}_{\mathfrak{A}_1}\times \mathfrak{S}_{\mathfrak{A}_2}$ is understood as a measure space obtained as a product of two measure spaces $\mathfrak{S}_{\mathfrak{A}_1}$ and $\mathfrak{S}_{\mathfrak{A}_2}$. A measure structure on 
$\mathfrak{S}_{\mathfrak{A}_i}$ is defined as the Borel structure determined by the corresponding weak-$*$ topology on $\mathfrak{S}_{\mathfrak{A}_i}$, $i =1,2$.

An arbitrary fixed decomposition of a state $\omega \in \mathfrak{S}_{\mathfrak{A}}$ corresponds to a measure $\mu$ such that
$\omega = \int_{\mathfrak S} \nu d\mu(\nu)$. As there are, in general, many decompositions (cf Section 7) we will be interested in measures from the following set
$$M_{\omega}(\mathfrak{S}_{\mathfrak{A}})\equiv M_{\omega} = \{ \mu : \omega = \int_{\mathfrak S} \nu d\mu(\nu) \},$$
i.e. the set of all Radon probability measures on $\mathfrak{S}_{\mathfrak A}$
 with the fixed barycenter $\omega$. Take an arbitrary measure $\mu$ from $M_{\omega}$. 
 By Theorem \ref{2.14} (cf also Remark \ref{2.15}) there exists a net of discrete measures (having a finite support) $\mu_k$ such that $\mu_k \to \mu$, and the convergence is understood in the  weak-$^*$ topology on $\mathfrak{S}_{\mathfrak{A}}$.
 
 Defining $\mu_1^k$ ($\mu_2^k$) analogously as $\mu_1$ ($\mu_2$ respectively; cf equations (\ref{71}), (\ref{72})),
 one has $\mu_1^k \to \mu_1$ and $\mu_2^k 
 \to \mu_2$, where again the convergence is taken in the weak-$^*$ topology on $\mathfrak{S}_{\mathfrak{A}_1}$ ($\mathfrak{S}_{\mathfrak{A}_2}$ respectively).
 To see this, note that $\mu_k \to \mu$ means that for any continuous function $f \in C(\mathfrak{S}_{\mathfrak A})$,
 \begin{equation}
 \label{trick}
  \mu_k(f) \to \mu(f).
  \end{equation} 
  But note, that $g \circ r_i$ is in
 $C(\mathfrak{S}_{\mathfrak A})$ for any $g \in C(\mathfrak{S}_{{\mathfrak A}_i})$. Thus, plugging
 $f= g \circ r_i$ in (\ref{trick}) one gets the convergence of $\mu_i^k$.
 
 Then define, for each $k$, $\boxtimes\mu^k$ as it was done in (\ref{73}).
 We can verify that $\left\{\boxtimes\mu^k\right\}$ is convergent (in weak $*$-topology) to a measure on  $\mathfrak{S}_{\mathfrak{A}_1}\times \mathfrak{S}_{\mathfrak{A}_2}$.
 To see this, take a continuous function $g$ on 
 $\mathfrak{S}_{\mathfrak{A}_1}\times \mathfrak{S}_{\mathfrak{A}_2}$. Observe that this two variable function gives rise to the following function $g(r_1 \cdot, r_2 \cdot) = \tilde{g}(\cdot)$, and obviously $\tilde{g}$ is a continuous function on $\mathfrak{S}_{\mathfrak{A}}$. Therefore,
$$\boxtimes\mu_k(g) = \left(\sum_{i=1}^{N_k} \lambda_{i,k} \delta_{r_1\rho_{i,k}} \times \delta_{r_2\rho_{i,k}}\right)(g)= \sum_{i=1}^{N_k} \lambda_{i,k} g(r_1\rho_{i,k},r_2\rho_{i,k})$$
$$= \sum_{i=1}^{N_k} \lambda_{i,k}\tilde{g}(\rho_{i,k}) = \left(\sum_{i=1}^{N_k} \lambda_{i,k} \delta_{\rho_{i,k}}\right)(\tilde{g}),$$
and the last term is convergent, by definition, to $\mu(\tilde{g})$.
 
 Consequently, taking the weak-$^*$ limit we arrive at the measure $\boxtimes \mu$ on  $\mathfrak{S}_{\mathfrak{A}_1}\times \mathfrak{S}_{\mathfrak{A}_2}$. It follows easily that $\boxtimes\mu$ does not depend on the chosen approximation procedure.

To grasp the idea which is behind the construction let us consider a very simple example:

\begin{exmp}
Let us fix a state $\omega$ and take a discrete, finite supported, measure $\mu_0$ in $M_{\omega}$; i.e.
$$\mu_0 = \sum_{i=1}^{N} \lambda_i \delta_{\rho_i},
$$
\end{exmp}\noindent
where $\lambda_i \geq0$ and $\sum_i \lambda_i = 1$. Note, that $\rho_i \in \rm{supp} \mu_0$ for any $i$.
Define, as before, $\boxtimes \mu_0 = \sum_i \lambda_i \delta_{r _1\rho_i}\times  \delta_{r _2\rho_i}$
and note that the measure $\boxtimes \mu_0$ on $\mathfrak{S}_{\mathfrak{A}_1}\times \mathfrak{S}_{\mathfrak{A}_2}$ defines the new state $\tilde{\omega}$ in the following way:
\begin{equation}
\tilde{\omega}(A_1\otimes A_2)=\int_{\mathfrak{S}_{\mathfrak{A}_1}\times \mathfrak{S}_{\mathfrak{A}_2}}{\varphi}(A_1\otimes A_2)(d\boxtimes\mu_0)(\varphi)
\end{equation}
where $\varphi \in \mathfrak{S}_{\mathfrak{A}_1}\times \mathfrak{S}_{\mathfrak{A}_2}$, i.e. $\varphi = (\varphi_1,\varphi_2)$. We have defined $\varphi(A_1 \otimes A_2)$ as $\varphi_1(A_1)\varphi_2(A_2)$
which can be considered as $(\varphi_1 \otimes \varphi_2)(A_1 \otimes A_2) \equiv \varphi(A_1 \otimes A_2)$.
Thus
\begin{equation}
\begin{split}
\tilde{\omega}(A_1\otimes A_2)&=\sum_{i=1}^N\lambda_i (r_1 \rho_i, r_2 \rho_i)(A_1\otimes A_2)\\
&=\sum_{i=1}^N\lambda_i(r_1 \rho_i)(A_1)(r_2 \rho_i)(A_2)\\
&=\sum_{i=1}^N\lambda_i (r_1 \rho_i)\otimes (r_2 \rho_i)(A_1\otimes A_2)
\end{split}
\end{equation}
Hence $\tilde{\omega}$ is a separable state which originates from the given one.\newline

Now, we are in position to give the definition of the coefficient of quantum correlations, $d(\omega, A_1,A_2) \equiv d(\omega, A) $, where $A_i \in \mathfrak{A}_i$.
\begin{defn}
Let a quantum composite system $\left(\mathfrak{A}=\mathfrak{A}_1\otimes \mathfrak{A}_2, \mathfrak{S}_{\mathfrak{A}} \right)$ be given. Take a
 $\omega\in \mathfrak{S}_{\mathfrak{A}}$. We define the coefficient of quantum correlations as
\begin{equation}
d(\omega, A)= \inf_{\mu\in M_{\omega}(\mathfrak{S}_{\mathfrak{A}})} \left|\int_{\mathfrak{S}_{\mathfrak{A}}}\xi(A) d\mu(\xi)-\int_{\mathfrak{S}_{\mathfrak{A}_1}\times \mathfrak{S}_{\mathfrak{A}_2}}\xi(A)(d\boxtimes\mu)(\xi)\right| \label{fin}
\end{equation}
\end{defn}

The formula (\ref{fin}) is a "measure" of extra non classical type of correlations. Namely, following the strategy of Kadison-Ringrose example, see Example \ref{exm6.1}, 
an evaluation of a distance between the given state $\omega$ and the set of approximative separable states is done.

It is a simple matter to see that 
$d(\omega, A)$ is equal to $0$ if the state $\omega$ is a separable one (cf arguments given prior to Theorem \ref{9.7}). The converse statement is much less obvious. However,
we are able to prove it. Namely:

\begin{thm}
\label{8.5}
Let $\mathfrak{A}$ be the tensor product of two C*-algebras $\mathfrak{A}_1$, $\mathfrak{A}_2$.
Then state $\omega \in \mathfrak{S}_{\mathfrak{A}}$ is separable if and only if $d(\omega, A)=0$ for all $A\in \mathfrak{A}_1\otimes \mathfrak{A}_2$
\end{thm}\noindent
\begin{proof}
The basic idea of the proof of the statement that $d(\omega, A)=0$ implies separability of $\omega$ 
relies on the study of continuity properties of the function
\begin{equation}
\label{gwiazdka3}
M_{\omega}(\mathfrak{S}_{\mathfrak A}) \ni \mu \mapsto \int_{\mathfrak{S}_{\mathfrak{A}}}\xi(A) d\mu(\xi)-\int_{\mathfrak{S}_{\mathfrak{A}_1}\times \mathfrak{S}_{\mathfrak{A}_2}}\xi(A)(d\boxtimes\mu)(\xi)
\end{equation}
and the proof falls naturally into few steps.
\begin{enumerate}
\item $M_{\omega}(\mathfrak{S}_{\mathfrak A})$ is a compact set.

We begin by recalling that $\mathfrak{S}_{\mathfrak A}$ is a compact set ($\mathfrak A$ has the unit
$\mathds{1}$) and the set of positive Radon measures $M^+(\mathfrak{S}_{\mathfrak A})$ is also compact (all in weak-$*$ topologies!). Take $\{\mu_{\alpha}\} \subset M_{\omega}(\mathfrak{S}_{\mathfrak A})$ such that $\mu_{\alpha} \to \mu$ (weakly).
But this implies
$$ \int \hat{A}(\varphi)d\mu_{\alpha}(\varphi) = \omega(A) \equiv \hat{A}(\omega) \to \int \hat{A}(\varphi)d\mu(\varphi)$$
Thus $\int \hat{A}(\varphi)d\mu(\varphi) = \omega(A)$. Hence $\mu \in M_{\omega}(\mathfrak{S}_{\mathfrak A})$. So $M_{\omega}(\mathfrak{S}_{\mathfrak A})$ being a closed subset of a compact set $M^+(\mathfrak{S}_{\mathfrak A})$ is a compact set.
\item The mapping $M_{\omega}(\mathfrak{S}_{\mathfrak A}) \ni \mu \mapsto \boxtimes\mu \in M^+(\mathfrak{S}_{{\mathfrak A}_1} \times \mathfrak{S}_{{\mathfrak A}_2})$ is weakly continuous.

To prove this, we should show that for any $\mu_0 \in M_{\omega}(\mathfrak{S}_{{\mathfrak A}})$
and any neighborhood $V(\boxtimes\mu_0; g_1,...,g_k)$ of $\boxtimes\mu_0$ there exists a neighborhood $U(\mu_0; f_1,...,f_k)$ of $\mu_0$ such that $\boxtimes\left(U(\mu_0; f_1,...,f_k)\right)\subseteq V(\boxtimes\mu_0; g_1,...,g_k)$ where $V\equiv V(\boxtimes\mu_0; g_1,...,g_k) = \{\boxtimes\mu: |\boxtimes\mu_0(g_i) - \boxtimes\mu(g_i)| < \epsilon, i=1,...,k \}$,
$g_i \in C(\mathfrak{S}_{{\mathfrak A}_1} \times \mathfrak{S}_{{\mathfrak A}_2})$ while
$U \equiv U(\mu_0; f_1,...,f_k) = \{ \mu: |\mu_0(f_i) - \mu(f_i)|< \epsilon_1, i=1,...,k \}$, $f_i \in C(\mathfrak{S}_{{\mathfrak A}})$.

The first step of the proof is to take
$\mu_0$ and $\mu$ in $M_{\omega}(\mathfrak{S}_{{\mathfrak A}})$ such that
\begin{equation}
\label{gwiazdka}
| \mu_0(f) - \mu(f)| < \epsilon \quad {\rm{for}} \quad f \in C(\mathfrak{S}_{\mathfrak A}).
\end{equation}
So, for simplicity, we put $k=1$ in the definition of neighborhoods $U$ and $V$.
Further, note that $f$ of the form
\begin{equation}
\label{gwiazdka2}
f(\rho) = g(r_1(\rho), r_2(\rho)) \quad \rho \in \mathfrak{S}_{\mathfrak A},
\end{equation}
where $g(\cdot, \cdot)$ is a continuous (two variables) function on $\mathfrak{S}_{{\mathfrak A}_1} \times \mathfrak{S}_{{\mathfrak A}_2}$ satisfing (\ref{gwiazdka}).

Let $\mu_0^n$ ($\mu^n$) be a weak-$^*$ Riemann approximation for $\mu_0$ ($\mu$ respectively). Then
$$|\mu_0^n(f) - \mu^n(f)| \leq |\mu_0^n(f) - \mu_0(f)| + |\mu_0(f) - \mu(f)| + |\mu(f) - \mu^n(f)| < \epsilon', $$
for all $f$ of the form (\ref{gwiazdka2}).

As a next step, let us consider a sequence $\boxtimes\mu_0^n$ ($\boxtimes\mu^n$) defining $\boxtimes\mu_0$ ($\boxtimes\mu$ respectively).
Note, that for any $f$ of the form (\ref{gwiazdka2}), one has
$$|\boxtimes\mu^n_0(f) - \mu_0^n(f)| = |\sum_{i=1}^{N} \lambda_{i,n} \delta_{r_1\rho_{i,n}} \times \delta_{r_2\rho_{i,n}}(f) - \sum_{i=1}^{N} \lambda_{i,n} \delta_{\rho_{i,n}}(f)|$$
$$=|\sum_{i=1}^{N} \lambda_{i,n} g(r_1(\rho_{i,n}), r_2(\rho_{i,n})) - \sum_{i=1}^{N} \lambda_{i,n} g(r_1(\rho_{i,n}), r_2(\rho_{i,n}))| = 0,$$
where $N < \infty$, and analogously for the second sequence. Therefore for any $f$ of the form (\ref{gwiazdka2}) one has
$$|\boxtimes\mu_0(g) -\boxtimes\mu(g)| \leq |\boxtimes\mu_0(g) - \boxtimes\mu_0^n(g)| + 
|\boxtimes\mu^n_0(g) - \mu_0^n(f)| +|\mu_0^n(f) - \mu_0(f)|$$
$$+|\mu_0(f) - \mu(f)| + |\mu(f) - \mu^n(f)| + |\mu^n(f) - \boxtimes\mu^n(g)| + |\boxtimes\mu^n(g) - \boxtimes\mu(g)| < 5 \epsilon,$$
for large enough $n$. Thus we have shown that for any $g \in C(\mathfrak{S}_{{\mathfrak A}_1} \times \mathfrak{S}_{{\mathfrak A}_2})$ 
\begin{equation}
|\boxtimes\mu_0(g) -\boxtimes\mu(g)| < 5 \epsilon,
\end{equation}
provided that $|\mu_0(f) - \mu(f)| < \epsilon$.
Therefore, if $V= \{\boxtimes\mu; |\boxtimes\mu_0(g_i) - \boxtimes\mu(g_i)| < \epsilon, i = 1,...,k\}$
with $g_i \in C(\mathfrak{S}_{{\mathfrak A}_1} \times \mathfrak{S}_{{\mathfrak A}_2})$ then there exists $U= \{ \mu; |\mu_0(f_i) - \mu(f_i)| <\frac{1}{5} \epsilon, i=1,...,k \}$ with $f_i$ of the form (\ref{gwiazdka2}) such that $\boxtimes(U) \subseteq V$.
But this means the continuity of the considered mapping.
\item The continuity proved in the second step implies that the function (\ref{gwiazdka3}) is a real valued, continuous function defined on a compact space. Hence, by Weierstrass theorem, infimum is attainable. Therefore, the condition $d(\omega, A) =0$ means that
\begin{equation}
\omega(A) = \int_{\mathfrak{S}_{\mathfrak{A}}} \xi(A) d\mu_0(\xi) = \int_{\mathfrak{S}_{\mathfrak{A}_1} \times \mathfrak{S}_{\mathfrak{A}_2}} \xi(A) d \boxtimes\mu_0(\xi),
\end{equation}
for all $A = A_1 \otimes A_2$. But, this means the separability of $\omega$.
\end{enumerate} 
\end{proof}
Theorem \ref{8.5} may be summarized by saying that any separable state contains ``classical'' correlations only. Therefore, \textbf{an entangled state contains ``non-classical'' (or pure quantum) correlations}.

To comment the question of separability of normal states we have two remarks:
\begin{rem}
\begin{enumerate}
\item (\textit{indirect way})

As we have considered $C^*$-algebra case, taking a normal state $\varphi \in \mathfrak{S}_{\mathfrak M}^n \equiv \mathfrak{S}_{\mathfrak M} \cap \mathfrak{M}_* \subset \mathfrak{S}_{\mathfrak M}$, we can apply Theorem \ref{8.5} for its analysis. If $d(\varphi, A)=0$ we are getting a ``separable'' decomposition of $\varphi$. However, still one must check whether components of the decomposition are normal or not. In other words, one must examine whether the measure providing the given decomposition is supported by $\mathfrak{S}_{\mathfrak M}^n$. It is worth pointing out that Theorem \ref{7.22} provides examples of measures being supported by $Ext(\mathfrak{S}_{\mathfrak M}^n)$
(if additionally the condition SC is satisfied).
\item (\textit{a possibility for a direct way})

One can try to modify the results obtained for $C^*$-algebra case to that which are relevant for $W^*$-algebra case. However, there are two essential differences. The first is given by Definition \ref{6.1} -- the closure of convex hull should be carried out with respect to the operator space projective norm topology. 

The second difference leads to a great problem. Namely $\mathfrak{S}_{\mathfrak M}^n$ is compact, in general, with respect to another topology than that which gives compactness of $\mathfrak{S}_{\mathfrak M}$. To illustrate this let us consider $\mathfrak{M} = B(\cal H)$, where $\cal H$ is an infinite dimensional Hilbert space. Then $\mathfrak{S}_{B(\cal H)}^n$ is a compact subset of $\cal{F}_T(\cal H)$ when it is equipped with $\sigma({\cal F}_T({\cal H}), {\cal F}_C(\cal H))$-topology. $\mathfrak{S}_{B(\cal H)}$ is compact with respect to $\sigma(B({\cal H})^*, B(\cal H))$-topology. Moreover, although the restriction $(r\omega)(A) = \omega(A \otimes \mathds{1})$, where $\omega \in \left(B({\cal H} \otimes B(\cal H)\right)^*$
is also well defined for a density matrix (it is given by the partial trace) the restriction $r$ is not, in general, $\sigma({\cal F}_T({\cal H} \otimes {\cal H}), {\cal F}_C({\cal H} \otimes {\cal H}))$ -- $\sigma({\cal F}_T({\cal H}), {\cal F}_C({\cal H}))$ continuous. As the continuity of the restriction map $r$ was crucial, the $C^*$-algebra case can not be straightforwardly modified.
\end{enumerate}
\end{rem}\noindent
We wish to end this section with another remark
\begin{rem}
Coefficient of quantum correlations yield information about quantum correlations. Therefore, it makes
legitimate study such correlation measures in an analysis of quantum stochastic dynamics. This question will be studied in the last section.
\end{rem}

\section {Entanglement of Formation}

$EoF$ - entanglement of formation is the second mathematical 'tool' for an analysis of entanglement that  we wish to consider. It is designed to separate separable states from entangled states and was introduced by Bennett, DiVincenzo, Smolin and Wooters in \cite{BDVSW} for finite dimensional case. The general definition of $EoF$ for quantum systems (so for infinite dimensional cases) appeared in \cite{Maj2002}.

The basic idea stems from the following observation. Let $\omega$ be a separable state on a quantum composite system 
specified by $\mathfrak{A} = \mathfrak{A}_1 \otimes \mathfrak{A}_2$. Decompose $\omega$ into pure states and apply the restriction map $r_1: \mathfrak{S}_{\mathfrak A} \to \mathfrak{S}_{\mathfrak{A}_1}$, given by $r_1\omega(A) = \omega(A\otimes\mathds{1})$, to each component of the decomposition. Let $\mathds{F}$ be a function defined on the set of states such that it takes the value $0$ only on pure states. Then, applying $\mathds{F}$ to the restriction of each component one gets an indicator of separability.

The important point to note here is that the restriction of a pure state is a pure one only for certain exceptional cases (cf Fact 
\ref{6.4}). To clarify this question we provide relevant results. The first one is (see Lemma 4.11 in \cite{takesaki}):

\begin{prop}
\label{9.1}
Let
$\mathfrak{A}=\mathfrak{A}_1\otimes \mathfrak{A}_2$, where $\mathfrak{A}_i$, $i=1,2$ is a $C^*$-algebra, and the state $\omega\in \mathfrak{S}_{\mathfrak{A}}$ be given. 
Denote by $r_1\omega$  a restriction of state $\omega$ to $\mathfrak{A}_1$ (we identify $\mathfrak{A}_1$ with $\mathfrak{A}_1\otimes \mathds{1}_2$).
Assume that $r_1\omega$ is a pure state. Then, $\omega(AB)=\omega(A)\omega(B)$ when $A\in \mathfrak{A}_1$ and
$B\in \mathfrak{A}_2$.
\end{prop}\noindent
Thus, the purity of $r_1\omega$ implies the factorization of $\omega$. The second result is (see Theorem 4.14 in \cite{takesaki}):

\begin{thm}
For two C*-algebras $\mathfrak{A}_1$ and $\mathfrak{A}_2$ the following conditions are equivalent
\begin{enumerate}
\item Either $\mathfrak{A}_1$ or $\mathfrak{A}_2$ is abelian
\item Every pure state $\omega$ on $\mathfrak{A}_1\otimes \mathfrak{A}_2$ is of the form $\omega=\omega_1\otimes \omega_2$ 
for some pure states $\omega_i$ on $\mathfrak{A}_i$, $i=1,2$.
\end{enumerate}
\end{thm}\noindent
Thus, the restriction of a pure state is a pure state for exceptional cases only. But, in this paper, we consider a quantum composite system. It means that both subsystems are quantum. Consequently, neither $\mathfrak{A}_1$ nor
$\mathfrak{A}_2$ is abelian.

To give the third result some preliminaries are necessary. Given a pair $(\mathfrak{A}, \omega)$ consisting of a $C^*$-algebra and a state, one can associate (via GNS construction, cf Example \ref{7.24}) the
Hilbert space $\mathcal{H}_{\omega}$ and the representation $\pi_{\omega}$. The family of all bounded linear operators on $\mathcal{H}_{\omega}$, as usually, will be denoted by  $B(\mathcal{H}_{\omega})$. The commutant of $\pi_{\omega}(\mathfrak{A})$ is defined as 
$\pi_{\omega}(\mathfrak{A})'=\left\{A\in B(\mathcal{H}_{\omega}); A\pi_{\omega}(B) = \pi_{\omega}(B)A \quad
\rm{for \ all} \ B \in \mathfrak{A}\right\}$.
In the same way one can define bicommutant $\pi_{\omega}(\mathfrak{A})''=\left(\pi_{\omega}(\mathfrak{A})'\right)'$. 
$\pi_{\omega}(\mathfrak{A})''$ is said to be a factor if $\pi_{\omega}(\mathfrak{A})''\cap\pi_{\omega}(\mathfrak{A})'=\left\{\mathds{C}\mathds{1}\right\}$.

\begin{defn}
A state $\omega$ on a C*-algebra $\mathfrak{A}$ is said to be factorial if $\pi_{\omega}(\mathfrak{A})''$ is a factor. 
\end{defn}\noindent
The promised, third result is (see Proposition 4.1.37 in \cite{BR}):

\begin{pro}
Let $\mathfrak{A}_1$, $\mathfrak{A}_2$ be C*-algebras and put $\mathfrak{A}=\mathfrak{A}_1\otimes \mathfrak{A}_2$. 
Denote by $r_1$ the restriction map $r_1:\mathfrak{S}_{\mathfrak{A}}\rightarrow \mathfrak{S}_{\mathfrak{A}_1}$. $r_1$ is  weak-$*$ continuous. Moreover,  $Ext(\mathfrak{S}_{\mathfrak{A}_1}) 
\subseteq r_1(Ext(\mathfrak{S}_{\mathfrak{A}}))\subseteq F_{\mathfrak{A}_1}$ where $F_{\mathfrak{A}_1}$ stands for 
factorial states on ${\mathfrak{A}_1}$, and $Ext(\mathfrak{ S}_{\mathfrak A})$ stands for the subset of extreme elements of $\mathfrak{ S}_{\mathfrak A}$. If $\mathfrak{A}_2$ is abelian then $Ext(\mathfrak{S}_{\mathfrak{A}_1}) 
=r_1(Ext(\mathfrak{S}_{\mathfrak{A}}))$.
\end{pro}

Consequently, if one considers the true quantum composite system, i.e. both subsystems are quantum (so both $C^*$-algebras $\mathfrak{A}_i$ are non commutative ones) then one can say only that the restriction of a pure state is a factorial one. This clearly indicates the role of the function $\mathds{F}$ in the definition given below.
Finally we note that (weak-$^*$)-(weak-$^*$) continuity of the restriction $r_1$ was already used in the discussion of quantum coefficient of correlations. We give:
\begin{defn}
Let $\omega$ be a state, $\omega\in F\subset \mathfrak{S}_{\mathfrak{A}_1\otimes \mathfrak{A}_2}$ and $F$ satisfy separability condition SC. The entanglement of formation $EoF$ is defined as
\begin{equation}
E_{\mathds{F}}(\omega)= \inf_{\mu\in M_{\omega}(\mathfrak{S}_{\mathfrak{A}_1\otimes \mathfrak{A}_2})}\int \mathds{F}(r\varphi) d\mu(\varphi) 
\end{equation}
where $\mathds{F}$ is a concave non-negative continuous function which vanishes on pure states and only on pure states, and to shorten notation we write $r$ instesd of $r_1$.
\end{defn}

Let us comment upon this definition. Firstly, we recall that a given state $\omega$ can have many decompositions, cf Section 7. Therefore,
we are forced to use the decomposition theory. In particular, orthogonal measures are playing an important role as they could be supported by $Ext(\mathfrak{S}_{\mathfrak{A}})$, see Theorem \ref{7.22} .
Further, we assumed the separability condition, SC, to avoid pathological measure-theoretical cases in the decomposition of $\omega$. But we remind the reader that this condition holds for all essential physical models -- see Section 7. Finally, in physical literature, one employs the von Neumann entropy as the function $\mathds{F}$. Namely

\begin{exmp}
Let $\mathfrak{A}_i=B(\mathcal{H}_i)$, $i=1,2$;  $\mathfrak{A}=\mathfrak{A}_1\otimes \mathfrak{A}_2$ and $F$ denote the set of all normal states on 
$B(\mathcal{H}_1)\otimes B(\mathcal{H}_2)=B(\mathcal{H}_1\otimes \mathcal{H}_2)$. 
Let $\varphi \in F$. Then $\varphi(\cdot) = Tr \rho_{\varphi} (\cdot)$, and one can identify $\varphi$ with $\rho_{\varphi}$. 
Take $\mathds{F}$ to be the von Neumann entropy
\begin{equation}
\mathds{F}(\varphi)=-\Tr\rho_{\varphi}\log\rho_{\varphi}
\end{equation}
Clearly, the von Neumann entropy satisfies all necessary properties provided that $\mathcal H$ is finite dimensional, see Section IID in \cite{Wehrl}. However, note that the function $\rho_{\varphi} \mapsto \Tr\{\rho_{\varphi}(\mathds{1} - \rho_{\varphi}) \}$ also possesses all necessary properties (again, for finite dimensional systems). Thus, this function can lead to another measure of entanglement (cf Section 11).
\end{exmp}

We wish to show that Entanglement of Formation can distinguish separable states from entangled states. To simplify notation, in the sequel, we will write $E(\omega)$ instead of $E_{\mathds{F}}(\omega)$.

$E(\omega)$ 
is defined as infimum of integrals evaluated on 
a continuous functions and the infimum is taken over the compact set. Therefore, the infimum is attainable (cf the previous section), i.e. 
there exists a measure $\mu_0\in M_{\omega}(\mathfrak{S})$ such that

\begin{equation}
E(\omega)=\int_{\mathfrak{S}}\mathds{F}(r\varphi) d\mu_0(\varphi)
\end{equation}
and, obviously,

\begin{equation}
\int_{\mathfrak{S}}\varphi d\mu_0(\varphi)=\omega
\end{equation}
Let us assume that $E(\omega)=0$, then we have
\begin{equation}
\int_{\mathfrak{S}}\mathds{F}(r\varphi) d\mu_0(\varphi)=0
\end{equation}
As $\mathds{F}(r\varphi)$ is non-negative, one can infer that
$\mathds{F}(r\varphi)=0$ on the support of $\mu_0$. 
But, as $\mathds{F}$ is a concave function, one has 
\begin{equation}
\mathds{F}(r\varphi)\geq \int_{\mathfrak{S}} \mathds{F}(rv)d\xi(v) 
\end{equation}
for any probability measure $d\xi$ on $\mathfrak{S}$ such
that $\varphi=\int_{\mathfrak{S}}vd\xi(v)$. In particular taking (as a measure $\xi$) 
a measure supported by pure states (from decomposition theory such measures exist) 
we conclude the existence of decomposition of $\varphi$ such that $\mathds{F}(rv)=0$ for $v$, hence $rv$ is a pure state and consequently $v$ is a product state, cf Proposition \ref{9.1}.
So, $\varphi$ is a convex combination of product states. Due to Theorem \ref{2.14} and Remark \ref{2.15},
 $\omega$ can be approximated  by a convex combination of product states. Consequently,  $\omega$ is a separable state.
\newline
\newline
Conversely, let $\omega$ be a separable state, i.e. 
\begin{equation}
\omega=\lim_{N \to \infty}\sum_{i=1}^N \lambda_i^{(N)} \omega_i^{(N)} 
\end{equation}
where each $\omega_i$ is a product state such that $\omega^{(N)}_i(A \otimes B) = \omega^{(N)}_{i,1}(A)\omega^{(N)}_{i,2}(B)$, where $\omega^{(N)}_{i,k}(\cdot)$ is a pure state on $\mathfrak{A}_k$. Define the sequence of measures $\mu^{(N)}$ in the following way:
\begin{equation}
\mu^{(N)}=\sum_{i=1}^N\lambda_i^{(N)} \delta_{\omega_i^{(N)}}
\end{equation}
where $\delta_{\omega_i^{(N)}}$ denotes Dirac's measure. 
If necessary, passing to a subsequence, we may suppose also that $\mu^{(N)}$ converges to $\mu \in M_{\omega}(\mathfrak{S}_{\mathfrak A})$ (it is always possible as $\{ \mu^{(N)} \} \subset M_1(\mathfrak{S}_{\mathfrak A})$, which a compact set).
Taking a weak limit of $\left\{\mu^{(N)}  \right\}$ one gets a measure $\mu$ such that

\begin{equation}
\int\varphi d\mu(\varphi)=\omega
\end{equation}
and
\begin{equation}
\int \mathds{F}(r\varphi) d\mu(\varphi)=0
\end{equation}
Thus, we arrived at:
\begin{thm}
\label{9.7}
$E(\omega)=0$ if and only if $\omega\in F$ is separable.
\end{thm}\noindent
It is worth pointing out that Entanglement of Formation, $EoF$, is not only a nice indicator of separability. It possesses also many useful properties like convexity, semi-continuity and others.
Here we will be concerned with convexity and with the property of $EoF$ which can be regarded as an analogue of entanglement witness. We begin with
\begin{prop}
$\mathfrak{S}_{\mathfrak A} \ni \omega \longrightarrow E(\omega)$
is a convex function.
\end{prop}
\begin{proof}
We have seen, Proposition \ref{7.13b}(2), that $\mu \in M_{\omega}(\mathfrak{S}_{\mathfrak A})$ if and only if $\mu(f) \geq f(\omega)$ for any real, continuous convex function on $\mathfrak{S}_{\mathfrak A}$. Thus, if $\mu_1 \in M_{{\omega}_1}(\mathfrak{S}_{\mathfrak A})$,
$\mu_2 \in M_{{\omega}_2}(\mathfrak{S}_{\mathfrak A})$, $\lambda_1\geq 0$, $\lambda_2\geq 0$, and $\lambda_1 + \lambda_2 = 1$ then
$$ (\lambda_1 \mu_1 + \lambda_2 \mu_2)(f) = \lambda_1 \mu_1(f) + \lambda_2 \mu_2(f) \geq \lambda_1 f(\omega_1) + \lambda_2 f(\omega_2) \geq f(\lambda_1\omega_1 + \lambda_2 \omega_2).$$
On the other hand $\lambda_1 \mu_1 + \lambda_2 \mu_2 \in M_{\lambda_1\omega_1 + \lambda_2 \omega_2}(\mathfrak{S}_{\mathfrak A})$ is equivalent to $(\lambda_1 \mu_1 + \lambda_2 \mu_2)(f)\geq f(\lambda_1 \omega_1 + \lambda_2 \omega_2)$. But, the convexity of $f$ implies $\lambda_1 f(\omega_1) + \lambda_2 f(\omega_2) \geq f(\lambda_1 \omega_1 + \lambda_2 \omega_2)$.
Therefore,  $\lambda_1 \mu_1 + \lambda_2 \mu_2 \in M_{\lambda_1 \omega_1 +  \lambda_2 \omega_2}(\mathfrak{S}_{\mathfrak A})$.
Consequently
\begin{equation}
\lambda_1 M_{\omega_1}(\mathfrak{S}_{\mathfrak A}) + \lambda_2 M_{\omega_2}(\mathfrak{S}_{\mathfrak A})\subseteq M_{\lambda_1 \omega_1 +  \lambda_2 \omega_2}(\mathfrak{S}_{\mathfrak A}).
\end{equation}
Hence
\begin{equation}
E(\lambda_1 \omega_1 + \lambda_2 \omega_2) = \inf_{\mu \in M_{\lambda_1 \omega_1 +  \lambda_2 \omega_2}(\mathfrak{S}_{\mathfrak A})}\int \mathds{F} \circ r (\varphi) d\mu(\varphi)
\end{equation}
$$\leq\lambda_1 \inf_{\mu \in M_{\omega_1}(\mathfrak{S}_{\mathfrak A})} \int \mathds{F} \circ (\varphi) d \mu(\varphi) + \lambda_2 \inf_{\mu \in M_{\omega_2}(\mathfrak{S}_{\mathfrak A})} \int \mathds{F} \circ (\varphi) d \mu(\varphi)
$$
$$= \lambda_1 E(\omega_1) + \lambda_2 E(\omega_2).$$
\end{proof}

As it was already announced, $EoF$ has a property which seems to be of the same nature as entanglement witness. To describe it we need some preliminaries. The first one is Bauer maximum principle (see Theorem 25.9 in \cite{choquet} or Lemma 4.1.12 in \cite{BR})

\begin{prop}
\label{9.9}
Let  $E$ be a Hausdorff locally convex topological space and $X\subset E$ a (non empty) convex compact subset. Suppose $f: X\longrightarrow \mathds{R}$ is convex and upper semi-continuous. Then there exists an extreme point of $X$ (not necessarily unique) at which $f$ assumes its maximum value.
\end{prop}

The second one is the characterization of extremal measures in $M_{\omega}(\mathfrak{S}_{\mathfrak A})$ (see Lemma 4.2.3 in \cite{BR}).

\begin{prop}
\label{9.10}
Let $\mathfrak A$ be a $C^*$-algebra (with identity) and $\omega \in \mathfrak{S}_{\mathfrak A}$.
Let $\mu$ be in $M_{\omega}(\mathfrak{S}_{\mathfrak A})$. Then, the following conditions are equivalent:
\begin{enumerate}
\item $\mu \in Ext\left(M_{\omega}(\mathfrak{S}_{\mathfrak A})\right)$.
\item the affine continuous functions over $\mathfrak{S}_{\mathfrak A}$ are dense in $L^1(\mathfrak{S}_{\mathfrak A}, \mu).$
\end{enumerate}
\end{prop}

Now we are ready to proceed with the description of the ``new entanglement witness''. To this end, we denote, for $\mu \in M_{\omega}(\mathfrak{S}_{\mathfrak A})$
$$ E_{\mu}(\omega) = \int (\mathds{F} \circ r) (\varphi) d \mu(\varphi).$$
It is clear that $ M_{\omega}(\mathfrak{S}_{\mathfrak A}) \ni \mu \mapsto E_{\mu}(\omega)$ is an affine, real valued, continuous function on $\mathfrak{S}_{\mathfrak A}$.
Moreover, obviously, $\min_x f(x) \geq A \Leftrightarrow\forall_x f(x) \geq A$. Hence, $\forall_x \quad -f(x) \leq - A$. Thus, $-A$ is maximum for the function ``$- f$''. Finally, as the set of all affine functions on $K$ is equal to $S(K)\cap \left(-S(K)\right)$, where $S(K)$ stands for the set of real continuous convex functions on $K$, an application of Proposition \ref{9.9} implies that minimum in definition of $E(\omega)$ (cf the proof of Theorem \ref{9.7}) is attained on a certain extremal measure $\mu_0$ in $M_{\omega}(\mathfrak{S}_{\mathfrak A})$. Then, applying Proposition \ref{9.10} one gets: there exists an affine function $\hat{A}_0 : \omega \mapsto \omega(A_0)$, where $A_0^* = A_0 \in \mathfrak A$ such that
\begin{equation}
|E(\omega) - \omega(A_0)| = |\int_{\mathfrak{S}_{\mathfrak A}} (\mathds{F} \circ r)(\varphi) d\mu_0(\varphi) - \int_{\mathfrak{S}_{\mathfrak A}} \hat{A}_0(\varphi) d\mu_0(\varphi)|
\end{equation}
$$ \leq
\int_{\mathfrak{S}_{\mathfrak A}}|(\mathds{F}\circ r)(\varphi) - \hat{A}_0(\varphi)|d\mu_0(\varphi) < \epsilon,
$$
for an arbitrary small $\epsilon$, as $\mathds{F} \circ r$ is a continuous function on a compact set $\mathfrak{S}_{\mathfrak A}$. Consequently, there exists an observable $A_0 = A^*_0 \in \mathfrak A$ such that its expectation value $\omega(A_0) \equiv <A_0>_{\omega}$ approximates $EoF$, $E(\omega)$, at a given state $\omega$.

We wish to close this section with some remarks concerning the continuity of $EoF$, i.e. of the mapping $\mathfrak{S}_{\mathfrak A} \ni \omega \mapsto E(\omega)$.
This mapping is a real valued, convex function defined on a compact set. It can be proved that it is lower semicontinuous. However, the proof of upper semicontinuity, Proposition 3 in \cite{Maj2002}, has a gap. We are greatly indebted to Maxim Shirokov for this observation. Thus, the function
$\mathfrak{S}_{\mathfrak A} \ni \omega \mapsto E(\omega)$ is a real, lower semi-continuous convex function only. Such functions are close to be continuous, see for example Theorem C, page 93, in \cite{RV} or Proposition 3.3 in \cite{Phelps}.
The problem is to find a (dense) open subset in the domain of the function. This is related to the existence of non-trivial interior in $\mathfrak{A}^*$ or in $\mathfrak{M}_*$. But, there is a problem. Namely,
any proper generating cone in a finite dimensional ordered Banach spaces has interior points. However, this is not true for infinite dimensional spaces.
In particular, although $\mathfrak{A}^+$ has a nontrivial interior, $\left(\mathfrak{A}^*\right)^+$ has not, for more details we refer the reader to \cite{BatR}. Consequently, continuity properties of $EoF$, $E(\omega)$ are of the same sort as those of quantum entropy, see \cite{Wehrl}.

\section{PPT states} 

In Quantum Computing, a characterization of states with positive partial transposition (PPT states) is of paramount importance. Namely, cf the end of Section 6, we have seen
\begin{equation}
\mathfrak{S}_{sep}\subset\mathfrak{S}_{PPT}\subset\mathfrak{S}.
\end{equation}

Thus, states in $\mathfrak{S}_{PPT}$ contain non-classical correlations. However, it is believed that genuine 
quantum correlations are contained in states from $\mathfrak{S}\setminus\mathfrak{S}_{PPT}$.
On the other hand, in early days of attempts to classify the structure of positive maps, Choi in \cite{Ch1}
observed that the transposition is playing a distinguished role in describing tensor product of matrices which are positive. Moreover, he noted that the answer is easy for dimension $2$ but very non-trivial for higher dimensions. For more recent account on the role of transposition in low dimensional matrix algebras see \cite{MT}.
This question was considered in Physics by Peres \cite{peres} and Horodecki's \cite{hor} and their research led to the first classification of quantum states in Quantum Information. Then, further investigations have shown the importance of PPT states for quantum computing, see \cite{hor2}.

Here, we wish to present a characterization of PPT states for quantum systems. To this end we will exploit links between tensor products
and properly chosen mapping spaces. This should be expected as for the definition of PPT states the theory of positive mappings was used. On the other hand, this theory is related to tensor products, and finally, the Grothendieck approach to tensor products is emerging from rules of quantization, see Section 5. In particular, we have seen, cf formulas (\ref{40}) and (\ref{41}):
$$
L(\mathfrak{A},\mathfrak{A})\cong (\mathfrak{A}\otimes_{\pi}\mathfrak{A}_*)^*,
$$
where $L(\mathfrak{A},\mathfrak{A})$ stands for the set of all bounded linear maps from $\mathfrak{A}$ to $\mathfrak{A}$,
and 
$$
\mathfrak{B}(\mathfrak{A},\mathfrak{A}_*)\cong L(\mathfrak{A},\mathfrak{A})
$$
where $\mathfrak{B}(\mathfrak{A},\mathfrak{A}_*)$ denotes the space of bounded bilinear forms on $\mathfrak{A} \times\mathfrak{A}_*$.
Let us specify the above results for $\mathfrak{A} \equiv B(\mathcal H)$. Thus, we will be concerned with Dirac's formalism of Quantum Mechanics.

We begin with, see \cite{stor} (the full account of the results below was given in \cite{Maj2014}) 
\begin{prop}
 There is an isometric isomorphism $\varphi\rightarrow \hat{\varphi}$ between $L(B(\mathcal{K}), B(\mathcal{H}))$  and $(B(\mathcal{K})\otimes_{\pi} \mathcal{F}_T(\mathcal{H}))^*$ such that

\begin{equation}
\hat{\varphi}(\sum_{i=1}^NA_i\otimes B_i)=\sum_{i=1}^NTr(\varphi(A_i)B_i^t)
\end{equation}
where  $B_i^t=\tau(B_i)$ ($\tau$ stands for the transposition), $\mathcal{F}_T(\mathcal{H})$ denotes the trace class operators on $\mathcal H$, and $\sum_{i=1}^NA_i\otimes B_i\in B(\mathcal{K})\otimes \mathcal{F}_T(\mathcal{H})$. Furthermore $\varphi\in L(B(\mathcal{K}), B(\mathcal{H}))^+$ if and only if $\hat{\varphi}$ is positive on $B(\mathcal{K})^+\otimes_{\pi} \mathcal{F}_T(\mathcal{H})^+$,
where $L(B(\mathcal{K}), B(\mathcal{H}))^+$ denotes those linear bounded maps which are positive, i.e $\varphi(A^*A) \geq 0$ for $A \in B(\mathcal{K})$.
\end{prop}

We emphasize that, here and in the reminder of this section, both Hilbert spaces $\mathcal H$ and $\mathcal K$ are, in general, \textit{infinite dimensional}.
As PPT states are ``dual'' to decomposable maps (see \cite{MM}, and last paragraphs of Section 6) we need an adaptation of the above Proposition for CP and co-CP maps. In \cite{MMO}, it was shown, where we have used the identification of $B(\mathcal H)_*$ with $\mathcal{F}_T(\mathcal H)$:
\begin{thm}
(1) Let $L(B(\mathcal{H}), B(\mathcal{K})_*)$ stands for the set of all linear bounded, normal ($*$-weak continuous) maps from $B(\mathcal{H})$ into $B(\mathcal{K})_*$. There is an isomorphism $\psi \longmapsto \hat{\psi}$ between $L(B(\mathcal{H}), B(\mathcal{K})_*)$ and $(B(\mathcal{H})\otimes B(\mathcal{K}))_*$ given by 

\begin{equation}
\hat{\psi}(\sum_{i=1}^NA_i\otimes B_i)=\sum_{i=1}^NTr_{\mathcal K}(\psi(A_i)B_i^t)
\end{equation}

The isomorphism is isometric if $\hat{\psi}$ is considered on $B(\mathcal{H})\otimes_{\pi} B(\mathcal{K})$.  Furthermore, $\hat{\psi}$ is positive on $(B(\mathcal{H})\otimes B(\mathcal{K}))^+$  if and only if $\psi$ is completely positive.\newline
(2) There is an isomorphism $\psi \longmapsto \hat{\psi} $ between $%
L\left({B}\left( \mathcal{H}\right),{B}(\mathcal{K})%
_{\ast }\right) $ and \newline $\left({B(\mathcal{H})\otimes B(\mathcal{K})}\right)_{\ast }$
given by 
\begin{equation}
\hat{\psi} \left( \sum_{i}A_{i}\otimes B_{i}\right) =\sum_{i}Tr_{\mathcal{K}}\psi
\left( A_{i}\right) B_{i},\text{\ \ }a_{i}\in {B}\left( \mathcal{H}%
\right) ,\text{ }b_{i}\in {B}\left( \mathcal{K}\right) .
\end{equation}%
This isomorphism is isometric if $\hat{\psi}$ is considered on ${B(H){\otimes}_{\pi} B(K)}$. 
Furthermore $\hat{\psi} $ is positive on $(B(\mathcal H)\otimes B(\mathcal K))^+$ if and only if  $\psi $ is complete co-positive.
\end{thm}

\begin{rem}
Since the projective norm $\| \cdot \|_{\pi}$ is submultiplicative, the involution $^*$ is isometric for this norm it follows that ${B(H){\otimes}_{\pi} B(K)}$ is a $^*$-Banach algebra. 
Consequently, the concept of positivity is well defined in ${B(H){\otimes}_{\pi} B(K)}$.
Therefore, in above Theorem, one can combine isometricity with positivity of functionals.
\end{rem}
This result is proving to be extremely useful in the study of so called entanglement mappings. Namely, following Belavkin-Ohya scheme, see \cite{BO1}, \cite{BO2}, let us take an arbitrary normal state $\omega$ on $B(\mathcal{H}) \otimes B(\mathcal{K})$. Thus, we fix a density matrix $\rho_{\omega}$ describing the composite state $\omega$.
Let its spectral decomposition be
\begin{equation}
\rho_{\omega} =\sum_{i=1}^N\lambda_i\left|e_i\right\rangle\left\langle e_i\right|
\end{equation}
where $\left\{e_i\right\}$ is an orthogonal system in $\mathcal{H}\otimes \mathcal{K}$. Define a map $T_{\xi}:\mathcal{K}\rightarrow \mathcal{H}\otimes \mathcal{K}$ by  
\begin{equation}
T_{\xi}\eta=\xi\otimes \eta
\end{equation}
Then, following Belavkin-Ohya scheme, we define the entanglement operator $H:\mathcal{H}\rightarrow\mathcal{H}\otimes \mathcal{K}\otimes \mathcal{K}$ by
\begin{equation}
H\eta=\sum_{i}\lambda_{i}^{\frac{1}{2}}\left(J_{\mathcal{H}\otimes \mathcal{K}}\otimes T^*_{J_{\mathcal{H}\eta}} \right)e_i\otimes e_i
\end{equation}
where $e_i \in \mathcal{H}\otimes \mathcal{K}$ for each $i$, and $J_{\mathcal{H}\otimes \mathcal{K}}$ is a complex conjugation defined by
\begin{equation}
J_{\mathcal{H}\otimes \mathcal{K}}\left(\sum_{i}(\dot{e_{i}},f) \dot{e_{i}} \right)=\sum_{i}\overline{(\dot{e_{i}},f)} \dot{e_{i}}
\end{equation}
where $\left\{\dot{e_{i}}\right\}$ is the extension (if necessary) of the orthonormal system $\left\{e_i\right\}$ to the complete orthogonal normal system $\left\{ \dot{e_{i}} \right \}$
in $\mathcal{H}\otimes \mathcal{K}$. 
($J_{\mathcal H}$ is defined analogously using the spectral resolution of $H^*H$). 
The definition of the entanglement operator $H$ is a necessary step to define the entanglement mapping $\phi: B(\mathcal{K}) \to B(\mathcal{H})_*$
\begin{equation}
\label{95}
\phi(B) = (H^*(\mathds{1} \otimes B)H)^t = J_{\mathcal H} H^* (\mathds{1} \otimes B)^* H J_{\mathcal H}.
\end{equation}
where $B \in B(\mathcal K)$.
Properties of the entanglement mapping  $\phi$ as well as its dual $\phi^*$ are contained in the next proposition, which was proved in \cite{MMO}.
\begin{prop}
\begin{enumerate}
\item The dual of the entanglement mapping $\phi^*: B(\mathcal{H}) \to B(\mathcal{K})_*$
has the following form
$$ \phi^*(A) = \Tr_{\mathcal{H} \otimes\mathcal{K}} H A^tH^*, \quad A \in B(\mathcal H).$$
\item
A state $\omega$ on $B(\mathcal{H}\otimes\mathcal{K})$ can be written as 
\begin{equation}
\omega(A\otimes B)= \Tr_{\mathcal{H}}A\phi(B)= \Tr_{\mathcal{K}}B\phi^*(A).
\end{equation}
\end{enumerate}
\end{prop}

The definition of PPT states is saying that any such state composed with the partial transposition is again a state. This fact 
combined with the above proposition leads to
\begin{thm}
\label{10.4}
PPT states are completely characterized by the mapping $\phi^*$ which is both CP and co-CP.
\end{thm}
Thus, taking an arbitrary fixed normal state $\omega$, one has its density matrix $\rho_{\omega}$. This density matrix gives rise to the entanglement mapping $\phi$. Then,  Theorem \ref{10.4} says that positivity properties of $\phi$
encodes PPT characterization of the given state $\omega$.

For further details on a characterization of PPT states we refer the reader to \cite{MMO}.

\section{Time evolution of quantum correlations - an example}

In the classical theory of particle systems one of the objectives is to produce, describe 
and analyze dynamical systems
with evolution originated from stochastic processes in such a way that their 
equilibrium states are Gibbs states (cf. \cite{Ligg}). 
A well known illustration is a number of papers describing the so 
called Glauber dynamics \cite{Glau}. To
perform a detailed analysis of dynamical system of that type, it is convenient 
to use the theory of Markov processes in
the context of (classical) $L_p$-spaces. 
In particular, for the Markov-Feller processes, using the unique correspondence
between a process and the corresponding dynamical semigroup one can give a recipe 
for a construction of Markov generators for this class of processes (see \cite{Ligg}).
In the middle of nineties, this program was carried out 
in the setting of quantum mechanics
\cite{R7}-\cite{R9}. It is worth noting 
that for this programme, \textit{specific algebras} mentioned in Rule 1 (see Section 4) stand for the quantum counterpart of stochastic variables.

Guided by the classical theory and applying 
generalized conditional expectations
(in the sense of Accardi-Cechnini), it was possible to define the corresponding 
Markov generators of the underlying
quantum Markov-Feller dynamics. In that way the quantum counterpart of the
classical recipe for the construction of quantum Markov generators was obtained.
In the sequel, quantum semigroups obtained in that framework will be called \textit{ quantum
stochastic dynamical semigroups}.

To sum up, the analysis given in \cite{R7}-\cite{R9} led to a general 
scheme for constructing quantum jump and diffusive
processes on a lattice. However,
it is natural to pose the following question: 
 Was the above described quantization made in a correct way?
 Namely, if the procedure took into account quantumness of dynamical maps, the obtained dynamics should enhance correlations.
 In particular, one should study the following situation: Let the evolution, given by a stochastic dynamical semigroup, starts from a separable state (so having only classical type of correlations).
 Then after some time this state should exhibit new type (so quantum) correlations. In other words, an entanglement should appear.
Therefore, to show that the quantization leading to quantum stochastic dynamical semigroups was done in a correct way we will study the production of entanglement. Here, we will restrict ourselves to a simple model. Namely, we will consider
 the jump-type dynamics $T_t$ of XXZ model on a finite lattice, see \cite{KM}. To this end we will work within the description of quantum composite systems, see Section 6 and we will use the recipe obtained in \cite{R7} for the infinitesimal generator of quantum stochastic dynamical semigroup.
 Performing
calculations similar to those given in the appendix of \cite{R7} one can show that for the considered dynamics $T_t$, its infinitesimal generator $\mathcal L$ has the form
\begin{equation}
\label{97}
\mathcal{L} = \mathcal{E} - \mathds{1},
\end{equation}\noindent
where the operator $\mathcal E$ takes
the following form

\begin{equation}
\label{97a}
\mathcal{E}(A)=\tau(\gamma^*A\gamma).
\end{equation}
Here
\begin{equation}
\gamma=\rho^{\frac{1}{2}}\big(\tau\rho\big)^{\frac{1}{2}}, \label{gamma_def}
\end{equation}
and $\rho$ stands for the Gibbs state while $\tau$ is defined as $\tau = \frac{1}{2} (\mathds{1} + \psi)$. $\psi$
describes a symmetry transformation associated with a region $\Lambda$. We recall that this form of $\tau$, allows us to describe flips for a lattice system.
An application of the relation between Schr$\ddot{\textrm{o}}$dinger and Heisenberg picture leads to the dual evolution 
$T^d_t$ of a state $\sigma$:
$$
\Tr(T^d_t(\sigma))A=\Tr\sigma T_t(A)
$$
for any state $\sigma$ and any observable $A$. In particular, one has
$$
\Tr(\mathcal{E}^d(\sigma)A)=\Tr\sigma\mathcal{E}(A)
$$
and
$$
\mathcal{E}^d(\sigma)=\gamma\tau(\sigma)\gamma^*
$$

As it was mentioned, we will be interested in a one dimensional quantum XXZ model. Thus, we will
consider a one-dimensional finite $\frac{1}{2}$-spin chain with $N+1$ sites indexed from $0$ to $N$ and the
corresponding algebra of observables generated by
$$
\sigma^{i_0}\otimes\sigma^{i_1}\otimes\ldots\otimes\sigma^{i_N}, 
$$
where $i_k\in\{0,1,2,3\}$, $k=0,\ldots N$, and $\sigma^j$, $j=0,1,2,3$ are Pauli matrices.

To write the explicit form of a symmetry transformation $\psi$ let us fix certain regions:  $\Lambda=\{0,1,\ldots,N\}$ and $\Lambda_I,\Lambda_{II}\subset\Lambda$, $\Lambda_I\cup\Lambda_{II}=\Lambda$,
$\Lambda_I\cap\Lambda_{II}=\emptyset$, Then $\mathcal{H}_{\Lambda_I}$  ($\mathcal{H}_{\Lambda_{II}}$) is $2^{|\Lambda_I|}$  ($2^{|\Lambda_{II}|}$ respectively)
dimensional Hilbert space, cf Example \ref{7.15}, ($| \Lambda|$ stands for the number of sites).\\
To specify the map $\psi$ we take a local transformation $\psi_{kl}$ defined as follows
$$
\psi_{kl}(A_1\otimes\ldots\otimes A_k\otimes\ldots\otimes A_l\otimes\ldots\otimes A_N)= \nonumber
$$
\begin{equation}
=A_1\otimes\ldots\otimes A_l\otimes\ldots\otimes A_k\otimes\ldots\otimes A_N \label{psiegz}
\end{equation}
which describes the exchange
between the sites, (so we put $\Lambda_I = \{k,l\}$, etc). 

The Hamiltonian of the XXZ system has the form:
$$
H=-\sum_{n=1}^N(\sigma^1_{n-1}\sigma^1_n+\sigma^2_{n-1}\sigma^2_n+\Delta\sigma^3_{n-1}\sigma^3_n)
$$
where $\sigma^j$, $j=1,2,3$ are Pauli matrices. Recall that the number $\Delta\neq 1$ is responsible for anisotropy of the model.
The corresponding Gibbs state is represented by the density matrix
$$
\rho=Z^{-1}\exp(-\beta H)
$$
where $Z=Tr(e^{-\beta H})$.

It was shown in Section 9 that $EoF$ can serve as a well defined measure of entanglement. Moreover, it was 
indicated that there is a freedom in taking the specific
form of $\mathds{F}$. Here, we put
$\mathds{F}(\rho) = S_L(\rho)=-\Tr(\rho(\rho-1))$. 
Clearly, the linear entropy $S_L$ satisfies the conditions discussed in Section 9. Thus
\begin{equation}
M(\sigma)=\inf_{\sigma=\sum_i\lambda_i\sigma_i}
\sum_i \lambda_i \Tr_2\left[(\Tr_1\sigma_i)-(\Tr_1\sigma_i)^2\right] \label{miara1}
\end{equation}
where infimum is taken over all decompositions of $\sigma$, provides a well defined measure of entanglement. We denote here the entanglement of formation by ``M'' as we wish to reserve the letter ``E'' for the entanglement production. 
Further, we note (cf. Section 9) that in (\ref{miara1}) it is enough to restrict oneself to decomposition of $\sigma$ 
into a convex combination of pure states. The main difficulty in calculation of (\ref{miara1}) is to carry out $\inf$ 
over all prescribed decompositions.

To overcome that problem (we follow arguments given in \cite{KM}) we begin with the simplified version of (\ref{miara1}):
\begin{equation}
M^a(\sigma)=Tr_2\left[(Tr_1\sigma)-(Tr_1\sigma)^2\right] \label{miara2}
\end{equation}
where again, $\sigma$ is the density matrix determining the considered state. We emphasize, that in (\ref{miara2}) the  $\inf$ was dropped. Clearly
\begin{equation}
M(\sigma)\leq M^a(\sigma) \label{ineq1}
\end{equation}
We shall use $M^a(\sigma)$ in the following measure of entanglement production:
\begin{equation}
E_a(T^d_t\rho)=M^a(T^d_t\rho)-M^a(\rho) \label{prod1}
\end{equation}
This formula, applied to small $t$ and a pure state $\rho$, leads to
\begin{equation}
E_a(T^d_t\rho)=M^a((1-t)\rho+t\mathcal{E}^d(\rho))-M(\rho)+ o(t^2), \label{prod2}
\end{equation}
where we have used: $T^d_t\rho=(1-t)\rho+t\mathcal{E}^d(\rho)+ao(t^2)$ ($o(t^k)=a_1t^k+a_2t^{k+1}+\ldots$) and the 
obvious fact that $M(\rho)=M^a(\rho)$ for a pure state.

It is not difficult to see that the true measure of entanglement production, $E$, satisfies
\begin{equation}
E(T^d_t\rho)=M(T^d_t\rho)-M(\rho)\leq E_a(T^d_t\rho) \label{prod4}
\end{equation}
Let us consider the following case: $\rho$ is a pure nonseparable state such that, for small $t$,
\begin{equation}
\frac{E_a(T^d_t\rho)}{t}\leq 0 \label{est1}
\end{equation}
Therefore, $E(T^d_t\rho)\leq 0$ and $M(T^d_t\rho)\leq M(\rho)$. Consequently, 
the negative sign of $\frac{E_a(T^d_t\rho)}{t}$ implies a decrease of entanglement.

As the next step, let us consider again pure nonseparable state $\rho$ such that, for small $t$, 
$\frac{E_a(T^d_t\rho)}{t}\geq 0$. To say something about $E(T^d_t\rho)$ for that case we 
discuss the difference between $M(T^d_t\rho)$ and 
$M^a((1-t)\rho+t\mathcal{E}^d(\rho))$, again for small $t$.

To this end, let us consider a (convex) decomposition $\sum\lambda_iP_i$ of 
$T^d_t\rho$ into pure states. Clearly
$$
||\sum\lambda_iP_i-((1-t)\rho+t\mathcal{E}^d(\rho))||\leq Ao(t^2)
$$
where $A\geq 0$ is a constant and we have used the given form of the infinitesimal generator of dynamics $T_t$, see (\ref{97}). As $Tr_1$ is a linear projection, 
the linear entropy $S_L$ is a continuous function,  then
$$
|M^a(\sum\lambda_iP_i)-M^a((1-t)\rho+t\mathcal{E}^d(\rho))|<A'o(t^2)
$$
for some positive constant $A'$  On the other hand,
$$
M^a\left(\sum\lambda_iP_i\right)=M^a(\rho+t(\mathcal{E}^d(\rho)-\rho)) + A^{''}o(t^2)
=M^a(\rho)+E_a(T^d_t\rho)
$$
Hence
$$
E_a(T^d_t\rho)\geq \sum\lambda_iM^a(P_i)- M^a(\rho)
$$
Furthermore
$$
E_a(T^d_t\rho)\geq M(T^d_t\rho)- M^a(\rho) \equiv M(T^d_t \rho) - M(\rho).
$$
Thus, positivity of $E_a(T^d_t\rho)$ allows a production of entanglement.

Finally, let us consider the case of $\rho$ being pure separable state. It was shown, see \cite{KM}:
$$
M^a(\mathcal{E}^d(\rho))=M(\mathcal{E}^d(\rho))=Tr_2\left(Tr_1(\mathcal{E}^d(\rho))(1-Tr_1(\mathcal{E}^d(\rho)))\right)
$$
Clearly $M^a(\mathcal{E}^d(\rho))>0$ shows that $\mathcal{E}^d(\rho)$ is entangled, hence $(1-t)\rho+t\mathcal{E}^d(\rho)$ is entangled. 
Consequently, for the map $\rho\rightarrow T^d_t(\rho)$ (for small $t$) one has entanglement production provided that $M^a(\mathcal{E}^d(\rho))$ is positive.

The above arguments, presented in \cite{KM}, were used as the starting point for computer simulations for spin chain consisting of five sites. The numerical results given in \cite{KM} clearly show the production of entanglement for pure separable states. Moreover, they also show that the production of entanglement is sensitive to changes of the temperature $\beta$ and the coefficient of anisotropy $\Delta$.	

\begin{proble}
Recently, an efficient method to compute some entaglement measures (for finite dimensional systems) was provided, see \cite{Geza}. One may ask whether this method allows to
carry out computer simulations for more general spin chains.
\end{proble}

\section{Conclusions}
We presented a concise scheme for quantization and analysis of one of the fundamental concepts of probability calculus -- the idea of correlations. Although, the quantization of coefficient of correlation is straightforward, the non-commutative setting offers new phenomenon. Namely, a state (which can be understood as a non-commutative integral) does not possess the weak$^*$- Riemann approximation property. This implies new type
correlations, which are called quantum correlations. To study this new type correlations we have used the decomposition theory as well as the selected results from the tensor product theory of metric spaces.
The utility of decomposition theory stems from the well known fact that the set of states in quantum mechanics does not form a simplex. Deep Grothendieck's results and their subsequent generalizations were necessary to provide the definition of separable (entangled) density matrices and to study the important class of states -- PPT states. In particular, note that the proper geometry for the collection of density matrices, for which people were looking for many years, is that given by the projective norm.

It is also important to note that our approach is applicable to real quantum systems, i.e. systems described by infinite dimensional structures. This is indispensable if one wishes to allow the canonical quantization.

Within such general approach to quantization, we gave detailed description of quantum correlations. In particular, the coefficient of quantum correlations was studied. This coefficient can be considered as a measure of quantumness of correlations. Moreover, a detailed exposition of entanglement of formation, $EoF$, was given in a general setting. We emphasize that our analysis of $EoF$ demonstrates rather strikingly that the value of $EoF$, for a fixed state, can be approximated with arbitrary precision by a mesurement of carefully selected observable.

We end the review with the example showing the evolution of entanglement  for a selected time evolution. This indicates, among others things, that the presented measures of quantum correlations are not restricted to the static case.  On the other hand, this does not mean that evolution of quantum correlations could be treated as causal relations.

Finally, this paper can be considered as an extension of the recent Spehner \cite{speh} survey for \textit{genuine} quantum systems.

\section{Acknowledgments}
This survey is the result of a series of lectures I gave on the subject of quantum correlations, firstly for PhD students in Gdansk University, Gdansk, January 2014 and secondly, in Summer School on Quantum Information in Zhejiang University, Hangzhou, July 7 -- August 11, 2014. I would like to thank the organizer of the Summer School, professor Wu Junde, and the participants for very nice School and wonderful hospitality. I would like also to thank M. Banacki  and A. Posiewnik for their help in editorial work as well as for suggested improvements and to David Blecher for his remarks on the operator space projective tensor products.

The author would like to express his hearty thanks to Marcin Marciniak and Tomasz Tylec who kindly read a preliminary draft of the paper and gave us truly valuable comments.
The support of the grant of Foundation for Polish Science TEAM
project cofinanced by the EU European Regional Development Fund
 is gratefully acknowledged.

\Addresses

\end{document}